\documentclass[a4paper,USenglish,cleveref,autoref,thm-restate]{lipics-v2021}
\title{Point Separation and Obstacle Removal by Finding and Hitting Odd Cycles}

\titlerunning{Algorithms for Point Separation and Obstacle Removal}

\author{Neeraj Kumar}{Department of Computer Science, University of California, Santa Barbara, USA} {neeraj@cs.ucsb.edu}{}{}
\author{Daniel Lokshtanov}{Department of Computer Science, University of California, Santa Barbara, USA} {daniello@ucsb.edu}{}{}
\author{Saket Saurabh}{IMSc, Chennai, India and  University of Bergen, Norway} {saket@imsc.res.in}{}{}
\author{Subhash Suri}{Department of Computer Science, University of California, Santa Barbara, USA} {suri@cs.ucsb.edu}{}{}
\author{Jie Xue}{New York University Shanghai, China} {jiexue@nyu.edu}{}{}
\authorrunning{N. Kumar et al.}
\Copyright{Neeraj Kumar, Daniel Lokshtanov, Saket Saurabh, Subhash Suri, and Jie Xue}

\nolinenumbers 

\EventEditors{John Q. Open and Joan R. Access}
\EventNoEds{2}
\EventLongTitle{42nd Conference on Very Important Topics (CVIT 2016)}
\EventShortTitle{CVIT 2016}
\EventAcronym{CVIT}
\EventYear{2016}
\EventDate{December 24--27, 2016}
\EventLocation{Little Whinging, United Kingdom}
\EventLogo{}
\SeriesVolume{42}
\ArticleNo{23}
\usepackage{xspace}
\usepackage{cite}

\def\mathify#1{\ifmmode{#1}\else\mbox{$#1$}\fi} 
\let\ensuremath=\mathify

\long\def\cut#1{{}}

\newtheorem{Conjecture}[lemma]{Conjecture}

\def\cb{\mathcal{B}}
\def\cs{\mathcal{S}}

\def\calr{\mathcal{R}}

\def\opt{\textit{OPT}}

\def\bbox{\mathcal{R}}
\def\col{\textsf{col}}
\def\cm{\mathcal{M}}
\def\Arr{\mathsf{Arr}}
\def\tone{\text{Type-1}\xspace}
\def\ttwo{\text{Type-2}\xspace}
\def\pointsep{\textsc{Points-separation}\xspace}
\def\obstacleremoval{\textsc{Obstacle-removal}\xspace}
\def\oddcycletrans{\textsc{Odd Cycle Transversal}\xspace}
\def\mid{\textit{mid}}
\def\refp{\textsf{ref}}
\def\lab{\textsf{lab}}
\def\dvrconj{\textsc{Dense vs Random}\xspace}
\def\gptsep{\textsc{Generalized Points-separation}\xspace}

\def\eps{\epsilon}

\newtheorem{fact}[lemma]{Fact}

\begin{document}
\maketitle

\begin{abstract}
    
    Suppose we are given a pair of points $s, t$ and a set $\cs$ of $n$ geometric objects in the plane, called obstacles.
    We show that in polynomial time one can construct an auxiliary (multi-)graph $G$ with vertex set $\cs$ and every edge labeled from $\{0, 1\}$, such that a set $\cs_d \subseteq \cs$ of obstacles separates $s$ from $t$ if and only if $G[\cs_d]$ contains a cycle whose sum of labels is odd. 
    Using this structural characterization of separating sets of obstacles we obtain the following algorithmic results.

    In the \obstacleremoval{} problem the task is to find a curve in the plane connecting $s$ to $t$ intersecting at most $q$ obstacles. 
     %
    %
    %
    We give a $2.3146^q n^{O(1)}$ algorithm for \obstacleremoval{}, significantly improving upon the previously best known $q^{O(q^3)}  n^{O(1)}$ algorithm of Eiben and Lokshtanov (SoCG'20). 
    %
    We also obtain an alternative proof of a constant factor approximation algorithm for \obstacleremoval{}, substantially simplifying the arguments of Kumar et al. (SODA'21).
    
    In the \gptsep{} problem input consists of the set $\cs{}$ of obstacles, a point set $A$ of $k$ points and $p$ pairs $(s_1, t_1), \ldots (s_p, t_p)$ of points from $A$.  The task is to find a minimum subset $\cs_r \subseteq \cs$ such that for every $i$, every curve from $s_i$ to $t_i$ intersects at least one obstacle in $\cs_r$.
    We obtain $2^{O(p)} n^{O(k)}$-time algorithm for \gptsep{}. This resolves an open problem of Cabello and Giannopoulos (SoCG'13), who asked about the existence of such an algorithm for the special case where $(s_1, t_1), \ldots (s_p, t_p)$ contains all the pairs of points in $A$.
    %
    %
    %
 %
    Finally, we improve the running time of our algorithm to $f(p,k) \cdot n^{O(\sqrt{k})}$ when the obstacles are unit disks, where $f(p,k) = 2^{O(p)} k^{O(k)}$, and show that, assuming the Exponential Time Hypothesis (ETH),  the running time dependence on $k$ of our algorithms is essentially optimal.
\end{abstract}

\section{Introduction}
\label{sec:intro}
Suppose we are given a set $\cs$ of geometric objects in the plane, and we want to modify $\cs$ in order to achieve certain guarantees on coverage of paths between a given set $A$ of points.
Such problems have received  significant interest in sensor networks~\cite{balister2009trap, KumarLA05, BeregK09, ChanK12}, robotics~\cite{lavalle, eiben2018improved} and 
computational geometry~\cite{Eiben2017, mcp-fpt, bksvCGTA}. 
There have been two closely related lines of work on this topic:
\textbf{(i)} \emph{remove} a smallest number of obstacles from $\cs$ to 
satisfy {\em reachability} requirements for points in $A$, and 
%
%
\textbf{(ii)} \emph{retain} a smallest number of obstacles to satisfy {\em separation} requirements for points in $A$.

In the most basic version of these problems the set $A$ consists of just two points $s$ and $t$. Specifically, in \obstacleremoval{} the task is to find a smallest possible set $\cs_d \subseteq \cs$ such that there is a curve from $s$ to $t$ in the plane avoiding all obstacles in $\cs \setminus \cs_d$.
In $2$-\pointsep{} the task is to find a smallest set $\cs_r \subseteq \cs$ such that every curve from $s$ to $t$ in the plane intersects at least one obstacle in $\cs_r$.
It is quite natural to require the obstacles in the set $\cs{}$ to be connected. Indeed, removing the connectivity requirements results in problems that are computationally intractable~\cite{Eiben2017,EibenK20,mcp2021}. 

When the obstacles are required to be connected \obstacleremoval{} remains \textsf{NP}-hard, but becomes more tractable from the perspective of approximation algorithms and parameterized algorithms.
For approximation algorithms, Bereg and Kirkpatrick~\cite{BeregK09} designed a constant factor approximation for unit disk obstacles. Chan and Kirkpatrick~\cite{ChanK12,ChanK14} improved the approximation factor for unit disk obstacles. Korman et al.~\cite{KormanLSS18} obtained a $(1 + \epsilon)$-approximation algorithm for the case when obstacles are fat, similarly sized, and no point in the plane is contained in more than a constant number of obstacles. 
Whether a constant factor approximation exists for general obstacles was posed repeatedly as an open problem~\cite{bksvCGTA,ChanK12,ChanK14} before it was resolved in the affirmative by a subset of the authors of this article~\cite{mcp2021}.

For parameterized algorithms, Korman et al.~\cite{KormanLSS18} designed an algorithm for \obstacleremoval{} with running time $f(q)n^{O(1)}$ for determining whether there exists a solution $\cs_d$ of size at most $q$, when obstacles are fat, similarly sized, and no point in the plane is contained in more than a constant number of obstacles. Eiben and Kanj~\cite{Eiben2017,EibenK20} generalized the result of Korman et al.~\cite{KormanLSS18}, and posed as an open problem the existence of a $f(q)n^{O(1)}$ time algorithm for \obstacleremoval{} with general connected obtacles. Eiben and Lokshtanov~\cite{mcp-fpt} resolved this problem in the affirmative, providing an algorithm with running time $q^{O(q^3)}n^{O(1)}$. 

Like \obstacleremoval{}, the $2$-\pointsep{} problem becomes more tractable when the obstacles are connected. Cabello and Giannopoulos~\cite{cabello2016complexity} showed that $2$-\pointsep{} with connected obstacles is polynomial time solvable. They show that the more general \pointsep{} problem where we are given a point set $A$ and asked to find a minimum size set $\cs{}_r \subseteq \cs$ that separates every pair of points in $A$, is \textsf{NP}-complete, even when all obstacles are unit disks. They leave as an open problem to determine the existence of $f(k)n^{O(1)}$ and $f(k)n^{g(k)}$ time algorithms for \pointsep{}, where $k = |A|$.

\paragraph*{Our Results and Techniques}
Our main result is a structural characterization of separating sets of obstacles in terms of odd cycles in an auxiliary graph. 

\begin{theorem}\label{thm:mainStructure}
There exists a polynomial time algorithm that takes as input a set $\cs$ of obstacles in the plane, two points $s$ and $t$, and outputs a (multi-)graph $G$ with vertex set $\cs$ and every edge labeled from $\{0, 1\}$, such that a set $\cs_d \subseteq \cs$ of obstacles separates $s$ from $t$ if and only if $G[\cs_d]$ contains a cycle whose sum of labels is odd. 
\end{theorem}

The proof of Theorem~\ref{thm:mainStructure} is an application of the well known fact that a closed curve separates $s$ from $t$ if and only if it crosses a curve from $s$ to $t$ an odd number of times. Theorem~\ref{thm:mainStructure} allows us to re-prove, improve, and generalize a number of results for \obstacleremoval{}, $2$-\pointsep{} and \pointsep{} in a remarkably simple way. More concretely, we obtain the following results. 

\begin{itemize}
\item {\em There exists a polynomial time algorithm for $2$-\pointsep{}.} 
\end{itemize}

Here is the proof: construct the graph $G$ from Theorem~\ref{thm:mainStructure} and find the shortest odd cycle, which is easy to do in polynomial time. This re-proves the main result of Cabello and Giannopoulos~\cite{cabello2016complexity}. Next we turn to  \obstacleremoval{}, and obtain an improved parameterized algorithm and simplified approximation algorithms. 

\begin{itemize}
\item {\em There exists an algorithm for \obstacleremoval{} that determines whether there exists a solution size set $\cs{}$ of size at most $q$ in time $2.3146^qn^{O(1)}$.}
\end{itemize}

Here is a proof sketch: construct the graph $G$ from Theorem~\ref{thm:mainStructure} and determine whether there exists a subset $\cs_d$ of ${\cal S}$ of size at most $q$ such that $G - \cs_d$ does not have any odd label cycle. This can be done in time $2.3146^qn^{O(1)}$ using the algorithm of Lokshtanov et al.~\cite{oct-fpt} for {\sc Odd Cycle Transversal}.\footnote{
The only reason this is a proof sketch rather than a proof is that the algorithm of Lokshtanov et al.~\cite{oct-fpt} works for unlabeled graphs, while $G$ has edges with labels $0$ or $1$. This difference can be worked out using a well-known and simple trick of subdividing every edge with label $0$ (see Section~\ref{sec:obs-removal}).}
This parameterized algorithm improves over the previously best known parameterized algorithm for \obstacleremoval{} of Eiben and Lokshtanov~\cite{mcp-fpt} with running time $q^{O(q^3)}n^{O(1)}$. 

If we run an approximation algorithm for {\sc Odd Cycle Transversal} on $G$ instead of a parameterized algorithm, we immediately obtain an approximation algorithm for \obstacleremoval{} with the same ratio. Thus, the $O(\sqrt{\log n})$-approximation algorithm for {\sc Odd Cycle Transversal}~\cite{AgarwalCMM05,kratsch2020representative} implies a $O(\sqrt{\log n})$-approximation algorithm for \obstacleremoval{} as well.
Going a little deeper we observe that the structure of $G$ implies that the standard Linear Programming relaxation of {\sc Odd Cycle Transversal} on $G$ only has a constant integrality gap. This yields a constant factor approximation for \obstacleremoval{}, substantially simplifying the approximation algorithm of Kumar et al~\cite{mcp2021}.

\begin{itemize}
\item {\em There exists a a constant factor approximation for \obstacleremoval{}.} 
\end{itemize}

Finally we turn our attention back to a generalization of \pointsep{}, called \gptsep{}. Here, instead of separating all $k$ points in $A$ from each other, we are only required to separate $p$ specific pairs $(s_1, t_1), \ldots, (s_p, t_p)$ of points in $A$ (which are specified in the input).
We apply Theorem~\ref{thm:mainStructure} several times, each time with the same obstacle set $\cs_{}$, but with a different pair $(s_i, t_i)$. Let $G_i$ be the graph resulting from the construction with the pair $(s_i, t_i)$.
Finding a minimum size set $\cs_r$ of obstacles that separates $s_i$ from $t_i$ for every $i$ now amounts to finding a minimum size set $\cs_r$ such that $G_i[\cs_r]$ contains an odd label cycle for every $i$.
The graph in the construction of Theorem~\ref{thm:mainStructure} does not depend on the points $(s_i, t_i)$ - only the labels of the edges do. Thus $G_1, \ldots, G_p$ are copies of the same graph $G$, but with $p$ different edge labelings. 
Our task now is to find a subgraph of $G$ on the minimum number of vertices, such that the subgraph contains an odd labeled cycle with respect to each one of the $p$ labels. 
We show that such a subgraph has at most $O(p)$ vertices of degree at least $3$ and use this to obtain a $2^{O(p^2)}n^{O(p)}$ time algorithm for \gptsep{}.
This implies a $2^{O(k^4)}n^{O(k^2)}$ time algorithm for \pointsep{}, resolving the open problem of Cabello and Giannopoulos~\cite{cabello2016complexity}. With additional technical effort we are able to bring down the running time of our algorithm for \gptsep{} to $2^{O(p)}n^{O(k)}$. This turns out to be close to the best one can do. On the other hand, for {\em pseudo-disk} obstacles we can get a faster algorithm.


\begin{itemize}
\item {\em There exists a $2^{O(p)}n^{O(k)}$ time algorithm for \gptsep{}, and a $n^{O(\sqrt{k})}$ time algorithm for \gptsep{} with pseudo-disk obstacles. }
\item {\em A $f(k)n^{o(k/\log k)}$ time algorithm for \pointsep{}, or a $f(k)n^{o(\sqrt{k})}$ time algorithm for \pointsep{} with pseudo-disk obstacles would violate the ETH~\cite{ImpagliazzoPZ01}.}
\end{itemize}

\section{Preliminaries}
\label{sec:prelims}

We begin by reviewing some relevant background and definitions.

\subparagraph*{Graphs and Arrangements} All graphs used in this paper are undirected. It will also be more convenient to sometimes consider multi-graphs, in which self-loops and parallel edges are allowed. The \emph{degree} of a vertex is the number of adjacent edges.

The \emph{arrangement} $\Arr(\cs)$ of a set of obstacles $\cs$ is a subdivision of the plane induced by the boundaries of the obstacles in $\cs$. The faces of $\Arr(\cs)$ are connected regions and edges are parts of obstacle boundaries. The \emph{arrangement graph} $G_\Arr = (V, E)$ is the dual graph of the arrangement whose vertices are faces of $\Arr(\cs)$ and edges connect neighboring faces. The complexity of the arrangement is the size of its arrangement graph which we denote by
$|\Arr(\cs)|$. We assume that the size of the arrangement is polynomial in the number of obstacles, that is $|\Arr(\cs)| = |G_\Arr| = n^{O(1)}$. This is indeed true for most reasonable obstacle models such as polygons or low-degree splines.

\subparagraph*{\obstacleremoval{} and \pointsep on Colored Graphs} Traditionally, \obstacleremoval{} problems have been defined in terms of graph problems on the
arrangement graph $G_\Arr$. In particular, we can define a \emph{coloring function} 
$\textsf{col} : V \rightarrow  2^{\cs}$ which assigns every vertex of $G_\Arr$ to the set of obstacles containing it. That is, obstacles correspond to colors in the colored graph $(G_\Arr, \col)$. It is easy to see that a curve connecting $s$ and $t$  in the plane that intersects $q$ obstacles corresponds to a path $\pi$ in the graph that uses $|\bigcup_{v \in \pi} \col(v)| = q$ colors in $(G_\Arr, \col)$ and vice versa.

We can also define 2-\pointsep as the problem of computing a \emph{min-color separator} of the graph $(G_\Arr, \col)$.
Let $V(\cs_r) \subseteq V$ be the set of vertices of $G_\Arr$ that contain at least one color from $\cs_r$.
A set of colors $\cs_r \subseteq \cs$ is a \emph{color separator} if $s$ and $t$ are disconnected in $G_\Arr - V(\cs_r)$. 
That is, every $s$--$t$ path must intersect at least one color in $\cs_r$. Therefore, a color separator of
minimum cardinality is a solution of 2-\pointsep, that is the minimum set of obstacles separating
$s$ from $t$.

The previous work~\cite{mcp2021} used structural properties of the colored graph $(G_\Arr, \col)$ to obtain a polytime algorithm for 2-\pointsep{} and a constant approximation for \obstacleremoval{}. One key difference in our approach is that instead of working on the colored graph $(G_\Arr, \col)$, we found it
more convenient to work with a so-called \emph{labeled intersection graph} $(G_\cs, \lab)$ of obstacles which we will formally construct in the next section.
Roughly speaking, given a set of obstacles $\cs$ and a \emph{reference curve} $\pi$ in the plane connecting $s$ and $t$, 
we build a multi-graph where vertices are obstacles in $\cs$ and edges connect a pair of intersecting obstacles. 
Every edge $e \in E$ is assigned a \emph{parity} label $\mathsf{lab}(e) \in \{0, 1\}$  based on the reference curve $\pi$. 
We say that a walk is labeled \emph{odd} (or \emph{even}) if the sum of labels of its edges is odd (or even) respectively.

Once this graph is constructed, we can forget about obstacles and formulate our problems using just the parity labels $\lab(e)$ on the edges of $G_\cs$. Since the parity function is much simpler to work with compared to the color function, this allows us to
significantly simplify the results from~\cite{mcp2021} and obtain new results.
In the next section, we describe the construction of graph $G_\cs$ and prove a key structural result that allow us to 
cast 2-\pointsep{} as finding shortest odd labeled cycle in $G_\cs$ and \obstacleremoval{} as the smallest \oddcycletrans of $G_\cs$. Recall that in \oddcycletrans problem, we want to find a set of vertices that ``hits'' (has non-empty intersection) with every odd-cycle of the graph. We will also need the following important property of plane curves.

\subparagraph{Plane curves and Crossings} A \textit{plane curve} (or simply \textit{curve}) is specified by a continuous function $\pi:[0,1] \rightarrow \mathbb{R}^2$, 
where the points $\pi(0)$ and $\pi(1)$ are called the \textit{endpoints} (for convenience, we also use the notation $\pi$ to denote the image of the path function $\pi$).
A curve is \textit{simple} if it is injective, and is \textit{closed} if its two endpoints are the same.
We say a curve $\pi$ \textit{separates} a pair $(a,b)$ of two points in $\mathbb{R}^2$ if $a$ and $b$ belong to different connected components of $\mathbb{R}^2 \backslash \pi$.

A \emph{crossing} of $\pi$ with $\pi'$ is an element of the set $\{t \in [0, 1] ~|~ \pi(t) \in \pi'\}$. 
%
We will often be concerned with the {\em number} of times $\pi$ crosses $\pi'$. This is defined as $|\{t \in [0, 1] ~|~ \pi(t) \in \pi'\}|$. Whenever we count the number of times a curve $\pi$ crosses another curve $\pi'$ we shall assume that (and ensure that) $|\{t \in [0, 1] ~|~ \pi(t) \in \pi'\}|$ is finite and that $\pi$ and $\pi'$ are {\em transverse}. That is for every $t \in [0, 1]$ such that $\pi(t) \in \pi'$ there exists an $\eps > 0$ such that the intersection of $\pi \cup \pi'$ with an $\eps$ radius ball around $\pi(t)$ is homotopic with two orthogonal lines. We will make frequent use of the following basic topological fact.


\begin{fact} \label{fact-separate}
Let $\pi$ be a curve with endpoints $a,b \in \mathbb{R}^2$. We have that
\begin{itemize}
    \item A \emph{simple closed curve} $\gamma$ separates $(a,b)$ iff $\pi$ crosses $\gamma$ an odd number of times.
    \item If $\pi$ crosses a \emph{closed curve} $\gamma$ an odd number of times, then $\gamma$ separates $(a,b)$.
\end{itemize}
\end{fact}

\subparagraph*{Partitions.} A \textit{partition} of a set $X$ is a collection $\varPhi$ of nonempty disjoint subsets (called \textit{parts}) of $X$ whose union is $X$.
For two partitions $\varPhi$ and $\varPhi'$ of $X$, we say $\varPhi$ is \textit{finer} than $\varPhi'$, denoted by $\varPhi \preceq \varPhi'$ or $\varPhi' \succeq \varPhi$, if for any $Y \in \varPhi$ there exists $Y' \in \varPhi'$ such that $Y \subseteq Y'$.
There is a one-to-one correspondence between partitions of $X$ and equivalence relations on $X$.
For any equivalence relation on a $X$, the set of its equivalence classes is a partition of $X$.
Conversely, any partition of $X$ induces a equivalence relation $\sim$ on $X$ where $x \sim y$ if $x$ and $y$ belong to the same part of the partition.
For two partitions $\varPhi$ and $\varPhi'$ of $X$, we define $\varPhi \odot \varPhi'$ as another partition of $I$ as follows.
Let $\sim_\varPhi$ and $\sim_{\varPhi'}$ be the equivalence relations on $X$ induced by $\varPhi$ and $\varPhi'$, respectively.
Define $\sim$ as the equivalence relation on $X$ where $x \sim y$ if $x \sim_\varPhi y$ and $x \sim_{\varPhi'} y$.
Then $\varPhi \odot \varPhi'$ is defined as the partition corresponding to the equivalence relation $\sim$.
Clearly, $\odot$ is a commutative and associative binary operation.
Thus, for a collection $\mathsf{Par}$ of partitions on $X$, we can define $\bigodot_{\varPhi \in \mathsf{Par}} \varPhi$ as the partition on $X$ obtained by ``adding'' the elements in $\mathsf{Par}$ using the operation $\odot$; note that $\bigodot_{\varPhi \in \mathsf{Par}} \varPhi$ is well-defined even if $\mathsf{Par}$ is infinite.

\begin{fact} \label{fact-represent}
Let $X$ be a set of size $k$ and $\varPhi_1,\dots,\varPhi_r$ be partitions of $X$.
Then there exists $T \subseteq [r]$ with $|T| < k$ such that $\bigodot_{t=1}^r \varPhi_t = \bigodot_{t \in T} \varPhi_t$.
\end{fact}
\begin{proof}
Let $T \subseteq [r]$ be a minimal subset satisfying $\bigodot_{t=1}^r \varPhi_t = \bigodot_{t \in T} \varPhi_t$.
We show $|T| < k$ by contradiction.
Assume $T = \{t_1,\dots,t_m\}$ where $m \geq k$.
Define $\varPsi_s = \bigodot_{i=1}^s \varPhi_{t_i}$ for $s \in [m]$.
Then we have $\varPsi_1 \succeq \cdots \succeq \varPsi_m$, which implies $1 \leq |\varPsi_1| \geq \cdots \geq |\varPsi_m| \leq k$.
It is impossible that $1 \leq |\varPsi_1| < \cdots < |\varPsi_m| \leq k$, because $m \geq k$.
Therefore, $\varPsi_s = \varPsi_{s+1}$ for some $s \in [m-1]$.
It follows that
\begin{equation*}
    \bigodot_{t \in T} \varPhi_t = \varPsi_{s+1} \odot \left(\bigodot_{i=s+2}^m \varPhi_{t_i}\right) = \varPsi_s \odot \left(\bigodot_{i=s+2}^m \varPhi_{t_i}\right) = \bigodot_{t \in T \backslash \{t_{s+1}\}} \varPhi_t,
\end{equation*}
which contradicts the minimality of $T$.
\end{proof}

\begin{fact} \label{fact-coarse}
Let $\varPhi$ be a partition of $X$ and suppose $|\varPhi| = z$.
For an integer $0 \leq d < z$, the number of partitions $\varPhi'$ satisfying $|\varPhi'| = z-d$ and $\varPhi' \succeq \varPhi$ is bounded by $z^{O(d)}$.
Furthermore, these partitions can be computed in $z^{O(d)}$ time given $\varPhi$.
\end{fact}
\begin{proof}
Consider the following procedure for generating a ``coarser'' partition from $\varPhi$.
We begin from the partition $\varPhi$.
At each step, we pick two elements $Y,Y'$ in the current partition and then replace them with their union $Y \cup Y'$ to obtain a new partition.
After $d$ steps, we obtain a partition $\varPhi'$ satisfying $|\varPhi'| = z-d$ and $\varPhi' \succeq \varPhi$.
Note that every partition $\varPhi'$ where $|\varPhi'| = z-d$ and $\varPhi' \succeq \varPhi$ can be constructed in this way.
Furthermore, the number of different choices at the $i$-th step is $\binom{z+1-i}{2} = O(z^2)$.
Therefore, the number of possible outcomes of the procedure, i.e., the number of partitions $\varPhi'$ satisfying $|\varPhi'| = z-d$ and $\varPhi' \succeq \varPhi$, is bounded by $z^{O(d)}$.
These partitions can be directly computed in $z^{O(d)}$ time via the procedure.
\end{proof}

\subparagraph*{Pseudo-disks.} A set $\mathcal{S}$ of geometric objects in $\mathbb{R}^2$ is a set of \textit{pseudo-disks}, if each object $S \in \mathcal{S}$ is topologically homeomorphic to a disk (and hence its boundary is a simple cycle in the plane) and the boundaries of any two objects $S,S' \in \mathcal{S}$ intersect at most twice.
Let $U$ be the union of a set $\mathcal{S}$ of pseudo-disks.
The boundary of $U$ consists of \textit{arcs} (each of which is a portion of the boundary of an object in $\mathcal{S}$) and \textit{break points} (each of which is an intersection point of the boundaries of two objects in $\mathcal{S}$).
We say two objects $S,S' \in \mathcal{S}$ \textit{contribute} to $U$ if an intersection point of the boundaries of $S$ and $S'$ is a break point on the boundary of $U$.
We shall use the following well-known property of pseudo-disks~\cite{kedem1986union}.
\begin{fact} \label{fact-planar}
Let $\mathcal{S}$ be a set of pseudo-disks, and $U$ be the union of the objects in $S$.
Then the graph $G = (\mathcal{S}, E)$ where $E = \{(S,S'): S,S' \in \mathcal{S} \textnormal{ contribute to } U\}$ is planar.
\end{fact}

\noindent
We remark that the above fact immediately implies another well-known property of pseudo-disks: the complexity of the union of a set of $n$ pseudo-disks is $O(n)$~\cite{kedem1986union}.
But this property will not be used in this paper.

\section{Labeled Intersection Graph of Obstacles}
\label{sec:lig}

We begin by describing the construction of the labeled intersection graph $G_\cs=(\cs, X)$ of the obstacles $\cs$. For the ease of exposition, we will use $S$ to refer to the obstacle $S \in \cs$ as well as the vertex for $S$ in $G_\cs$ interchangeably.

\subparagraph*{\boldmath Constructing the graph $G_\cs$} For every obstacle $S \in \cs$ we first select an arbitrary point $\refp(S) \in S$ and designate it to be the \emph{reference point} of the obstacle. 
Next, we select the \emph{reference curve} $\pi$ to be a simple curve in the plane connecting $s$ and $t$ such that
including it to the arrangement $\Arr(\cs)$ does not significantly increase its complexity.
That is, we want to ensure that $|\Arr(\cs \cup \pi)| = O(|\Arr(\cs)|)$.
Additionally, the reference curve $\pi$ is chosen such that there exists an $\epsilon > 0$ and $\pi$ is disjoint from an $\epsilon$ ball around every intersection point of two obstacles in $\Arr({\cal S})$ and from an $\epsilon$ ball around every reference point $\refp(S)$ for $S \in {\cal S}$.

As long as the intersection of every pair of obstacles is finite and their arrangement has
bounded size,
a suitable choice for $\pi$ always exists (and can be efficiently computed). For example one can choose $\pi$ to be the plane curve corresponding to an $s$--$t$ path in $G_\Arr$.

We will now add edges to $G_\cs$ as follows. (See also Figure~\ref{fig:odd-cycle}(c) for an example.)

\begin{itemize}
    \item For every obstacle $S \in \cs$ that contains $s$ or $t$, add a self-loop $e = (S, S)$ with $\lab(e) = 1$.
    \item For every pair of obstacles $S, S' \in \cs$ that intersect, we add edges to $G$ as follows.
    \begin{itemize}
        \item Add an edge $e_0 = (S, S')$ with $\lab(e_0) = 0$ if there exists a curve connecting $\refp(S)$ and $\refp(S')$ contained in the region $S \cup S'$ that crosses $\pi$ an \emph{even} number of times.
        \item Add an edge $e_1 = (S, S')$ with $\lab(e_1) = 1$ if there exists a curve connecting $\refp(S)$ and $\refp(S')$ contained in the region $S \cup S'$ that crosses $\pi$ an \emph{odd} number of times.
    \end{itemize}
\end{itemize}

Checking whether there exists a curve contained in the region $S \cup S'$ with endpoints $\refp(S)$ and $\refp(S')$ 
that crosses $\pi$ an odd (resp. even) number of times can be done in time linear in the size of arrangement 
$\Arr' = \Arr(S \cup S' \cup \pi)$. Specifically, we build the arrangement graph $G_{\Arr'}$ and only retain edges
$(f_i, f_j)$ such that the faces $f_i, f_j \in S \cup S'$. If the common boundary of faces $f_i, f_j$ is a portion of
$\pi$, we assign a label $1$ to the edge $(f_i, f_j)$, otherwise we assign it a label $0$. 
An odd (resp. even) labeled walk in $G_{\Arr'}$ connecting the faces containing $\refp(S)$ and $\refp(S')$ gives us the desired plane curve $\pi_{ij}$. Since edges of $G_{\Arr'}$ connect adjacent faces of $\Arr'$, we can ensure that the intersections between curve $\pi_{ij}$ and
the edges of arrangement (including parts of reference curve $\pi$) are all transverse.

We are now ready to prove the following important structural property of the graph $G_\cs$.

\begin{lemma}
    \label{lem:odd-cycle}
    A set of obstacles $\cs' \subseteq S$ in the graph $G_\cs$ separates the points $s$ and $t$ if and only if 
    the induced graph $H = G_\cs[\cs']$ contains an odd labeled cycle.
\end{lemma}

\begin{proof}
    $(\Rightarrow)$ For the forward direction, suppose we are given a set of obstacles $\cs'$ that separate $s$ from $t$.
    If $s$ or $t$ are contained in some obstacle, then we must have an odd self-loop in $G_\cs$ and we will be done.
    Otherwise, assume that $s, t$ lie in the exterior of all obstacles, so we have $s, t \not \in \calr(\cs')$ where 
    $\calr(\cs') = \bigcup_{S \in \cs'} S$ is the region bounded by obstacles in $\cs'$.
    Observe that $s, t$ must lie in different connected regions $R_s, R_t$ of $\mathbb{R}^2 \setminus \calr(\cs')$ or else 
    the set $\cs'$ would not separate them.
    At least one of $R_s$ or $R_t$ must be bounded, wlog assume it is $R_s$. Let $\gamma'$ be the simple closed 
    curve that is the common boundary of $\calr(\cs')$ and $R_s$. We have that $\gamma'$ encloses $s$ but not $t$
    and therefore separates $s$ from $t$. Using first statement of Fact~\ref{fact-separate}, 
    we obtain that $\gamma'$ crosses the reference curve $\pi$ an odd number of times. 
    Observe that the curve $\gamma'$ consists of multiple \emph{sections} $\alpha_1' \rightarrow \alpha_2' \dots \rightarrow \alpha_r'$ where each curve $\alpha_i'$ is part of the boundary of some obstacle $S_i$. 
    For each of these curves $\alpha_i'$, we add a \emph{detour} to and back from the reference point $\refp(S_i)$ of the obstacle it belongs. Specifically, let $q_i$ be an arbitrary point on the curve $\alpha_i'$ and let $\alpha_{i\ell}', \alpha_{ir}'$ be the portion of $\alpha_i'$ before and after $q_i$ respectively. We add the \emph{detour curve} $\delta_i = q_i \rightarrow \refp(S_i) \rightarrow q_i$ ensuring that it always stays within the obstacle $S_i$ which is possible because the obstacles are connected. (Same as before the curve $\delta_i$ can be 
    chosen to be transverse with $\pi$ by considering the corresponding walk in graph of $\Arr(S_i \cup \pi)$.)
    Let $\alpha_i = \alpha_{i\ell}' \rightarrow \delta_i \rightarrow \alpha_{ir}'$ be the curve obtained by adding detour $\delta_i$ to $\alpha_i'$.
    Let $\gamma = \alpha_1 \rightarrow \alpha_2 \dots \rightarrow \alpha_r$ be the closed curve obtained by adding these detours to $\gamma'$. Note that $\gamma$ is not necessarily simple as the detour curves may intersect each other. Every detour $\delta_i$ consists of identical copies of two curves, so it crosses the reference curve $\pi$ an even number of times. Since $\gamma'$ crosses $\pi$ an odd number of times, the curve $\gamma$ also crosses $\pi$ an odd number of times. (See also Figure~\ref{fig:odd-cycle}.) Observe that $\gamma$ and $\gamma'$ are transverse with $\pi$ because intersections of $\pi$ and obstacle boundaries are transverse and the detour curves $\delta_i$ are chosen to be transverse with $\pi$.
   
    \begin{figure}[htb!]
	\centering
	\includegraphics{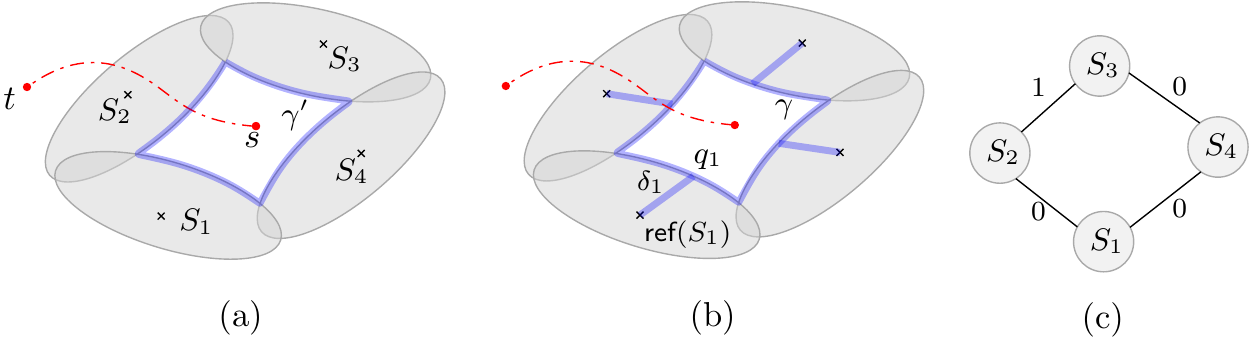}
	\caption{(a) The curve $\gamma'$ shown shaded in blue is the common boundary of $\calr(\cs')$ and region $R_s$ ~(b) Adding detours $\delta_i$ to obtain curve $\gamma$
    (c) Labeled Intersection graph $G_\cs$ ob obstacles}
	\label{fig:odd-cycle}
	\end{figure}
	
    We will now translate the curve $\gamma$ to a \emph{walk} in the labeled intersection graph $G_\cs$.
    Specifically, consider the section of $\gamma$ between two consecutive detours: $\gamma_{i,i+1} = \refp(S_i) \rightarrow q_i \rightarrow q_{i+1} \rightarrow \refp(S_{i+1})$.
    Therefore the obstacles $S_i, S_{i+1}$ must intersect and we have a curve $\gamma_{i,i+1}$
    connecting their reference points contained in the region $S_i \cup S_{i+1}$ that also intersects the reference curve $\pi$ an odd (resp. even) number of times. By construction, $G_\cs$ must contain an edge $e_{i, i+1}$ with label $1$ (resp. $0$).  By replacing all these sections of $\gamma$ with the corresponding edges of $G_\cs$, we obtain an odd-labeled closed walk $W$ in $G_\cs$. Of all the odd-labeled closed sub-walks of $W$, we select one that is inclusion minimal. This gives a simple odd-labeled cycle in $G_\cs[\cs']$.

    $(\Leftarrow)$ The reverse direction is relatively simpler. Given an odd-labeled cycle in $G_\cs[\cs']$, we obtain a closed curve $\gamma$ in the plane contained in region $\calr(\cs')$ as follows. For every edge $e_i = (S, S')$ of the cycle with label $\lab(e_i)$, we consider the curve $\gamma_i$ that connects the reference points $\refp(S)$ and $\refp(S')$ contained in $S \cup S'$ and crosses the reference curve 
    $\pi$ consistent with $\lab(e_i)$. Moreover $\gamma_i$ needs to be transverse with $\pi$.
    Such a curve exists by construction of $G_\cs$ .
    Combining these curves $\gamma_i$ in order gives us a closed curve $\gamma$ in the plane that crosses $\pi$ an odd number of times. Although this curve may be self intersecting, from second statement of Fact~\ref{fact-separate}, we have that $\gamma$ separates $s$ and $t$.
\end{proof}

The construction of the graph $G_\cs$, together with Lemma~\ref{lem:odd-cycle} prove Theorem~\ref{thm:mainStructure}.

\subparagraph{\boldmath 2-\pointsep{} as Shortest Odd Cycle in $G_\cs$} From Lemma~\ref{lem:odd-cycle}, it follows that a minimum set of obstacles that separates $s$ from $t$
corresponds to an odd-labeled cycle in $G_\cs$ with fewest vertices. This readily gives a polytime algorithm for 2-\pointsep{}. In particular, for a fixed starting vertex, we can compute the shortest odd cycle in $G_\cs$ in $O(|\cs|^2)$ time by the following well-known technique.
Consider an unlabeled auxiliary graph $G'$ with vertex set is $\cs \times \{0, 1\}$. For every edge $e = (S, S')$ of $G_\cs$, we add edges $\{(S, 0), (S', 0)\}$ and $\{(S, 1), (S',1)\}$ if $\lab(e) = 0$. Otherwise, we add the edges $\{(S, 0), (S', 1)\}$ and $\{(S, 1), (S',0)\}$. The shortest odd cycle containing a fixed vertex $S$ is the shortest path in $G'$ between vertices $(S,0)$ and $(S,1)$. Repeating over all starting vertices gives the shortest odd cycle in $G_\cs$. This can be easily extended for the node-weighted case which gives us
the following useful lemma that also yields a polynomial time algorithm for 2-\pointsep{}, reproving a result of Cabello and Giannopoulos~\cite{cabello2016complexity}.
\begin{lemma}
    \label{lem:separation-oracle}
    There exists a polynomial time algorithm for computing a minimum weight labeled odd cycle in the graph $G_\cs$.
\end{lemma}

Next we prove one more structural property of labeled intersection graph $G_\cs$ that will be useful later. We define a (labeled) \emph{spanning tree} $T$ of a connected labeled multi-graph $G_\cs$ to be a subgraph of $G_\cs$ that is a tree and connects all vertices in $\cs$. An edge $e = (u, v) \in G_\cs$ is a \emph{tree edge} if $(u, v) \in T$, otherwise it is called a \emph{non-tree} edge. 
\begin{lemma}
    \label{lem:tree-edge-walk}
    Let $G_\cs$ be a connected labeled intersection graph and $T$ be a spanning tree of $G_\cs$. If $G_\cs$ contains an odd labeled cycle, then it also contains an odd labeled cycle with exactly one non-tree edge.
\end{lemma}
\begin{proof}
    Let $C$ be an odd cycle in $G_\cs$ that contains fewest non-tree edges. If $C$ consists of exactly one non-tree edge, we are done. Otherwise, $C$ contains more than one non-tree edge. Let $e = (u, v) \in C$ be a non-tree edge and $C' \subset C$ be the remainder of $C$ without the edge $e$. Since $C$ is odd labeled, we must have $\lab(C') \neq \lab(e)$.
    
    Let $\pi_{uv}$ be the unique path connecting $u, v$ in $T$. This gives us a path $\pi_{uv}$ with label $\lab(\pi_{uv})$. Recall that $\lab(C') \neq \lab(e)$. We have two cases.
     {\em (i)}  If $\lab(\pi_{uv}) \neq \lab(e)$, then we obtain an odd labeled cycle $\pi_{uv} \oplus e$ that has one non-tree edge, namely $e$, and we are done.
        {\em (ii)}  Otherwise, $\lab(\pi_{uv}) = \lab(e) \neq \lab(C')$. This gives us
        an odd labeled closed walk $W^* = \pi_{uv} \oplus C'$ which contains one less non-tree edge than $C$. Let $C^* \subseteq W^*$ be an odd-labeled inclusion minimal closed sub-walk of $W^*$ (one such $C^*$ always exists). Therefore, $C^*$ is an odd-labeled cycle in $G_\cs$ that has fewer non-tree edges than $C$. But $C$ was chosen to be an odd labeled cycle with fewest non-tree edges, a contradiction.
\end{proof}
The above lemma also gives a simple $O(S^2)$ algorithm to \emph{detect} whether there exists an odd label cycle in $G_\cs$. Specifically, consider an arbitrary spanning tree of $T$ of $G_\cs$ and for
each edge not in $T$, compare its label with the label of the path connecting its endpoints in $T$.
\begin{lemma}
Given a labeled graph $G_\cs$, there exists an $O(S^2)$ time algorithm to detect whether $G_\cs$
contains an odd labeled cycle.
\end{lemma}

\section{Application to \obstacleremoval{}}
\label{sec:obs-removal}
We will show how to cast \obstacleremoval{} as a Labeled \oddcycletrans problem on the graph $G_\cs$.
Recall that in \obstacleremoval{} problem, we want to remove a set $\cs_d \subseteq \cs$
of obstacles from the input so that $s$ and $t$ are connected in $\cs \setminus \cs_d$. Equivalently, we want to select a subset $\cs_d$ of obstacles such that the complement set $\cs \setminus \cs_d$ \emph{does not
separate} $s$ and $t$. From Lemma~\ref{lem:odd-cycle}, it follows that the obstacles $\cs \setminus \cs_d$ do not separate $s$ and $t$ if and only if $G_\cs[\cs \setminus \cs_d]$ does not contain
an odd labeled cycle. 
This gives us the following important lemma.
\begin{lemma}
    \label{lemma:obstacle-removal}
    A set of obstacles $\cs_d \subseteq \cs$ is a solution to \obstacleremoval{} if and only if the set of vertices $\cs_d$ is a solution to \oddcycletrans of $G_\cs$. 
\end{lemma}
This allows us to apply the set of existing results for \oddcycletrans{} to obstacle removal problems. In particular, this readily gives an improved algorithm for \obstacleremoval when parameterized by the solution size (number of removed obstacles).
Let $G_\cs^+$ denote the graph $G_\cs$ where every edge $e$ with $\lab(e) = 0$ is subdivided. Clearly an odd-labeled cycle in $G_\cs$ has
odd length in $G_\cs^+$ and vice versa.
Applying the FPT algorithm for \oddcycletrans from~\cite{oct-fpt} on the graph $G_\cs^+$ gives us the following result.
\begin{theorem}
    There exists a $2.3146^k n^{O(1)}$ algorithm for \obstacleremoval{} parameterized by $k$, the number of removed obstacles.
\end{theorem}

This also immediately gives us an $O(\sqrt{\log n})$ approximation for \obstacleremoval{} by using the best known $O(\sqrt{\log n})$-approximation~\cite{agarwal2005log} for 
on the graph $G_\cs^+$. Observe that instances of obstacle removal are special cases of odd cycle transversal, specifically where the graph $G_\cs$ is an intersection graph of obstacles.
By applying known results on \emph{small diameter decomposition} of \emph{region intersection graphs}, Kumar et al.~\cite{mcp2021} obtained a constant factor approximation for \obstacleremoval. In the next section we present an alternative constant factor approximation algorithm. Although our algorithm follows a similar high level approach of using small diameter decomposition of $G_\cs$, we give an alternative proof of the approximation bound which significantly simplifies the arguments of~\cite{mcp2021}.

\subsubsection*{Constant Approximation for \obstacleremoval{}}
Our algorithm is based on formulating and rounding a standard LP for
labeled odd cycle transversal on labeled intersection graph $G_\cs$.
Let $0 \leq x_i \leq 1$ be an indicator variable that denotes whether obstacle $S_i$ is included to the solution or not. The LP formulation which will be referred as \textsc{Hit-odd-cycles-LP} can be written as follows:
\begin{align*}
    \min &\sum_{S_i \in \cs} x_i \\
    \text{subject to:} &~\\
    &\sum_{S_j \in C} x_j ~\geq~ 1 &&\text{for all odd-labeled cycles $C \in G_\cs$}
\end{align*}

Although this LP has exponentially many constraints, it can be solved in polynomial
time using ellipsoid method with the polynomial time algorithm for minimum weight odd cycle
in $G_\cs$ (Lemma~\ref{lem:separation-oracle}) as separation oracle. 
The next step is to round the fractional solution $\hat{x} = {x_1, x_2, \dots, x_n}$ obtained from solving the \textsc{Hit-odd-cycles-LP}. We will need some background on
small diameter decomposition of graphs.
\subparagraph*{Small Diameter Decomposition} Given a graph $G = (V, E)$ and a distance function $d : V \rightarrow \mathbb{R^+}$ associated with each vertex, we
can define the distance of each edge as $d(e) = d(v) + d(w)$ for
every edge $e= (v,w) \in E$. We can then extend the distance function to
any pair of vertices $d(u, v)$ as the shortest path distance between $u$ and $v$ in the edge-weighted graph with distance values of edges as edge weights.
We use the following result of Lee~\cite{lee2016separators} for the special case of \emph{region intersection graph} over planar graphs. 
\begin{lemma}
    \label{lemma:rig}
    Let $G = (V, E)$ be a node-weighted intersection graph of connected
    regions in the plane, then there exists a set $X \subseteq V$ 
    of $ |X| = O(1/\Delta) \cdot \sum d(v)$ vertices such that the diameter
    of $G-X$ is at most $\Delta$ in the metric $d$. Moreover, such a set $X$
    can be computed in polynomial time.
\end{lemma}
For the sake of convenience, we assume that $G_\cs$ does not contain an obstacle $S_i$ with a self-loop, because if so, we must always include $S_i$ to the solution.
Let $G_\cs^*$ be the underlying unlabeled graph obtained by removing labels
and multi-edges from $G_\cs$. Since  $G_\cs^*$ is simply the intersection graph of
connected regions in the plane, it is easy to show that $G_\cs^*$ 
is a region intersection graph over a planar graph (See also Lemma~4.1~\cite{mcp2021} for more details.) 

\subparagraph*{(Algorithm: \textsc{Hit-Odd-Cycles})} With small diameter decomposition for $G_\cs^*$ in place, the rounding algorithm is really simple. 
\begin{itemize}
    \item Assign distance values to remaining vertices of $G_\cs^* = (\cs \setminus \cs_0, E)$ as $d(S_i) = x_i$, where $x_i$ is the fractional solution obtained from solving
    \textsc{Hit-Odd-Cycle-LP}.
    \item Apply Lemma~\ref{lemma:rig} on graph $G_\cs^*$ with diameter $\Delta = 1/2$. Return the set of vertices $X$ obtained from applying the lemma as solution.
\end{itemize}

It remains to show that the set $X \subseteq \cs$ returned above indeed hits all the odd labeled cycles in $G_\cs$. Define a ball $\cb(c, R) = \{ v \in V : d(c, v) < R - d(v)/2 \}$ with center $c$, radius $R$ and distance metric $d$ defined before. Intuitively, $\cb(c, R)$ consists of the vertices that lie strictly inside the radius $R$ ball drawn
with $c$ as center.
\begin{lemma}
    The set $X$ returned by algorithm \textsc{Hit-Odd-Cycles} hits all odd labeled cycles in $G_\cs$. 
\end{lemma}
\begin{proof}
 The proof is by contradiction. Let $C$ be an odd labeled cycle such that $C \cap X = \emptyset$. Then $C$ must be contained in a single connected $\kappa$ component of $G_\cs - X$. Let $v_1$ be an arbitrary vertex of $C$ and consider a ball $B = \cb(v_1, 1/2)$ of radius $1/2$ centered at $v_1$. We have $\kappa \subseteq B$ due to the choice of diameter $\Delta$.
 Consider the shortest path tree $T$ of ball $B$ rooted at $v_1$ using the distance function $d(e)$ in the unlabeled graph $G_\cs^*$. For every edge $(u, v) \in T$ assign the label $\lab(e)$ of $ e = (u, v) \in G_\cs$. If  multiple labeled edges exist between $u$ and $v$,
 choose one arbitrarily.
 
 Now consider the induced subgraph $G_\cs' = G_\cs[B]$ which is a connected labeled intersection graph of obstacles in the ball $B$. Moreover, $T$ is a spanning tree of
 $G_\cs'$, and $G_\cs'$ contains an odd-labeled cycle because $\kappa \subseteq G_\cs'$.
 Applying Lemma~\ref{lem:tree-edge-walk} gives us an odd-labeled cycle $C \in G_\cs'$
 that contains exactly one edge $e \not\in T$. The cost of this cycle is $\textsf{cost}(C) < 1/2 + 1/2 = 1$. This contradicts the constraint of
 \textsc{Hit-Odd-Cycle-LP} corresponding to $C$.
\end{proof}

We conclude with the main result for this section.
\begin{theorem}
    There exists a polynomial time constant factor approximation algorithm for \obstacleremoval.
\end{theorem}

\section{\boldmath A Simple Algorithm for Generalized \pointsep{}}
\label{sec:pointsep-overview}
So far, we have focused on separating a pair of points $s, t$ in the plane. In this section, we consider the more general problem where we are given a set $\cs$ of $n$ obstacles, 
a set of points $A$ and a set and $P = \{(s_1, t_1), \dots, (s_p, t_p)\}$ of $p$ pairs of points in $A$
which we want to separate. First we show how to extend the labeled intersecting graph $G_\cs$ to $p$ source-destination pairs and that the optimal solution subgraph $G_\cs[\cs_\opt]$ exhibits a `nice' structure. Then we exploit this structure to obtain an $2^{O(p^2)} n^{O(p)}$ exact algorithm for \gptsep. Since $p = O(k^2)$, this algorithm runs in polynomial time for any fixed $k$, resolving an open question of~\cite{cabello2016complexity}. Using a more sophisticated approach, we later show how to improve the running time to $2^{O(p)} n^{O(k)}$. 

Recall the construction of the labeled intersection graph $G_\cs$ for a single point pair
$(s, t)$ from Section~\ref{sec:lig}. The \emph{label} $\lab(e) \in \{0, 1\}$ of each edge $e \in G_\cs$ denotes the \emph{parity} of edge $e$ with respect to \emph{reference curve} $\pi$ connecting $s$ and $t$. As we generalize the graph $G_\cs = (\cs, E)$ to $p$ point pairs, we extend the label function $\lab : E \rightarrow \{0, 1\}^p$ as a $p$-bit binary string that denotes the parity with 
respect to reference curve $\pi_i$ connecting $s_i$ and $t_i$ for all $i \in [p]$.
We will use $\lab_i(e)$ to denote the $i$-th bit of $\lab(e)$.

\subparagraph*{Generalized Label Intersection Graph:} 
    \begin{itemize}
    \item For each $(s_i, t_i) \in P$ and each $S \in \cs$ that contains at least one of $s_i$ or $t_i$, we add a self loop $e$ on $S$ with $\lab_i(e) = 1$ and $\lab_j(e) = 0$ for all $j \neq i$.
    \item For every pair of intersecting obstacles $S, S'$ and a $p$-bit string $\ell \in \{0, 1\}^p$:
    \begin{itemize}
        \item Let $\Pi = \{\pi_i ~|~ s_i, t_i \not\in S \cup S'\}$ be the set of reference curves that do not have endpoints in $S \cup S'$.
        \item We add an edge $e = (S, S')$ with $\lab(e) = \ell$ if there exists a plane curve connecting $\refp(S)$ and $\refp(S')$ contained in $S \cup S'$ that crosses all reference curves $\pi_i \in \Pi$ with parity consistent with label $\ell$.
        That is, the curve crosses $\pi_i$ and odd (resp. even) number of times if $i$-th bit of $\ell$ is $1$ (resp. zero).
    \end{itemize}
 \end{itemize}

Similar to the one pair case, we can build an unlabeled graph $G'$ with vertex set $\cs \times \{0,1\}^p$ and edges between them based on the arrangement $\Arr(S \cup S' \cup \bigcup \pi_i)$. Using this graph, we can obtain the following lemma. The proof is the
same as that of Lemma~\ref{lem:graph-construction}, with $p$ bit labels instead of $k$ bit labels.
\begin{lemma}
 The generalized labeled graph $G_\cs$ with $p$-bit labels can be constructed in $2^{O(p)} n^{O(1)}$ time.
\end{lemma}
Suppose we define $G_\cs(i)$ to be the image of $G_\cs$ induced by the labeling $\lab_i : E \rightarrow \{0, 1\}$. Specifically, we obtain $G_\cs(i)$ from $G_\cs$ by replacing label of each edge by the $i$-th bit $\lab_i(e)$, followed by removing parallel edges that have the same label. Observe that $G_\cs(i)$ is precisely the graph obtained by applying algorithm from Section~\ref{sec:lig} with reference curve $\pi_i$.

We say that a subgraph $G'_\cs \subseteq G_\cs$ is \emph{well-behaved} if $G_\cs'(i)$
contains an odd labeled cycle for all $i \in [p]$. We have the following lemma that can be obtained by applying Lemma~\ref{lem:odd-cycle} for every pair $(s_i, t_i) \in P$.

\begin{lemma}
\label{lem:well-behaved}
A set of obstacles $\cs' \subseteq \cs$ separate all point pairs in $P$ iff $G_\cs[\cs']$ is well-behaved.
\end{lemma}

We will prove the following important property of well-behaved subgraphs of $G_\cs$.

\begin{lemma}
    \label{lem:bounded-connectors}
    Let $G \subseteq G_\cs$ be an inclusion minimal well-behaved subgraph of $G_\cs$.
    Then there exists a set $V_c \subseteq V(G)$ of \emph{connector vertices} such that $G$ consists of the vertex set $V_c$ and a set of $K$ \emph{chains} (path of degree 2 vertices) with endpoints in $V_c$. Moreover, $|V_c| \leq 4p$ and $|K| \leq 5p$.
\end{lemma}
\begin{proof}
    Since $G$ is inclusion minimal well-behaved subgraph, it does not contain
    a proper subgraph that is also well-behaved. Therefore, $G$ does not
    contain a vertex of degree 
    at most $1$ because such vertices and edges adjacent
    to them cannot be part of any cycle. Suppose $G$ has $r$ connected components $C_1,\dots,C_r$. We fix a spanning tree $T_j$ of $C_j$ for each $j \in [r]$.
    We construct the set $V_c$ by including every vertex of degree three or more
    to $V_c$. The components $C_j$ that do not contain a vertex of degree three
    must be a simple cycle because $G$ does not have degree-1 vertices. For every
    such $C_j$, we include vertices adjacent to the only non-tree edge of $C_j$.
    It is easy to verify that $G$ consists of $K$ chains connecting vertices
    in $V_c$.
    
    Let $E_0$ be the set of non-tree edges, that are edges not in $T_j$ for some $j \in [r]$. We claim that $|E_0| \leq p$. Since $G$ is well-behaved, $G(i)$ consists an odd-labeled cycle for all $i \in [p]$. Using Lemma~\ref{lem:tree-edge-walk}, and the spanning tree $T_j$ of the component containing that odd labeled cycle, we can transform into an odd-labeled cycle that uses at most one non-tree edge. Repeating this for all pairs, we can use at most $p$ edges from $E_0$. If $|E_0| > p$, then we would have a proper subgraph of $G$ with at most $p$ edges that is also well-behaved, which is not possible because $G$ was chosen to be inclusion minimal. Therefore $|E_0| \leq p$.

    The graph $G$ only contains vertices of degree $2$ or higher, hence each leaf node of the trees $T_1, \dots, T_r$ must be adjacent to some edge in $E_0$. Therefore, the number of leaf nodes is at most $2p$, and so the number of nodes of degree three or above in $T_1, \dots, T_r$ is
    also at most $2p$. Observe that the vertices in $V_c$ are either adjacent to some edge in $E_0$ or have degree three or more in some tree $T_j$. The number of both these type of vertices is at most $2p$, which gives us $|V_c| \leq 4p$.
    Finally, we bound $|K|$, the number of chains. Note that each edge of $G$ belongs to exactly one chain in $K$. Therefore, the number of chains containing at least one edge in $E_0$ is at most $p$, because $|E_0| \leq p$. All the other chains that do not have any edge in $E_0$, are contained in the trees $T_1,\dots,T_r$. It follows that these chains do not form any cycle, and thus their number is less than $|V_c|$. This gives us $|K| \leq 5p$.
\end{proof}

It is easy to see that if $\cs' \subseteq \cs$ is an optimal set of obstacles separating all pairs in $P$, then there exists an inclusion minimal well-behaved subgraph $G$ of $G_\cs[\cs']$ that
satisfies the property of Lemma~\ref{lem:bounded-connectors}.
Observe that the $K$ chains of graph $G$ are vertex disjoint, so for every chain $K_t$ connecting
vertices $S_i, S_j \in V_c$ that has $\lab(K_t) = \ell$, an optimal solution will always choose the walk in $G_\cs$ that has label $\ell$ and has fewest vertices.
To that end, we will need the following simple lemma which is a generalization of algorithm to compute shortest odd cycle in $G_\cs$ with 1-bit labels.
\begin{lemma}
    \label{lem:shortest-labeled-path}
    Given a labeled graph $G_\cs =(\cs, E)$ with labeling $\lab: E \rightarrow \{0,1\}^p$, 
    the shortest walk between any pair of vertices $S_i, S_j$ with a fixed label $\ell \in \{0,1\}^p$ can be computed in $2^{O(p)} n^{O(1)}$ time.
\end{lemma}

\subparagraph*{Algorithm: \textsc{Separate-Point-Pairs}}
\begin{enumerate}
    \item For every pair of vertices $S_i, S_j \in \cs$ and every label $\ell \in \{0, 1\}^p$,
    precompute the shortest walk connecting $S_i, S_j$ with label $\ell$ in $G_\cs$ using Lemma~\ref{lem:shortest-labeled-path}.
    \item For all possible sets $V_c \subseteq \cs$ and ways of connecting $V_c$ by $K$ chains:
    \begin{itemize}
        \item For all $(2^p)^{5p} = 2^{O(p^2)}$ possible labeling of $K$ chains:
        \begin{enumerate}
            \item Let $G \subseteq G_\cs$ be the labeled graph consisting of vertices $V_c$ and chains $K_t \in K$ replaced by shortest walk between endpoints of $K_t$ with label $\lab(K_t)$, already computed in Step 1.
            \item Check if the graph $G$ is well-behaved. If so, add its vertices as one candidate solution.
        \end{enumerate}
    \end{itemize}
    \item Return the candidate vertex set with smallest size as solution.
\end{enumerate}

Precomputing labeled shortest walks in Step 1 takes at most $2^{O(p)} n^{O(p)}$ time.
The total number of candidate graphs $G$ is $n^{O(p)} \cdot p^{O(p)} \cdot 2^{O(p^2)}$, and checking if it is well behaved can be done in $n^{O(1)}$ time. We have the following result.

\begin{theorem}
\textnormal{\gptsep} for connected obstacles in the plane can be solved in $2^{O(p^2)} n^{O(p)}$ time, where $n$ is the number of obstacle and $p$ is the number of point-pairs to be separated.
\end{theorem}

\begin{corollary}
\textnormal{\sc Point-Separation} for connected obstacles in the plane can be solved in $2^{O(k^4)} n^{O(k^2)}$ time, where $n$ is the number of obstacles and $k$ is the number of points. This is polynomial in $n$ for every fixed $k$.
\end{corollary}

\section{\boldmath A Faster Algorithm for Generalized \pointsep{}}
\label{sec:algorithm}
\noindent
Recall that the labeled graph $G_\cs$ constructed in the previous section consisted of labels that are $p$-bit binary strings. As a result, the running time has a dependence of $n^{O(p)}$ which in worst case could be $n^{O(k^2)}$, for example, in the case of \pointsep when $P$ consists of all point pairs.
In this section, we describe an alternative approach that builds a labeled intersection graph whose labels are $k$-bit strings. Using this graph and the notion of \emph{parity partitions}, we obtain an $2^{O(p)} n^{O(k)}$ algorithm for \gptsep which gets rid
of the $n^{O(k^2)}$ dependence for \pointsep. 
The construction of graph $G_\cs$ is almost the same as before, except that now we choose the {reference curves} $\pi_i$ differently. In particular, let $A = \{a_1, a_2, \dots, a_k\}$ be the set of points and $P$ be a set of pairs $(a_i, a_j)$ of points we want to separate. We pick an arbitrary point $o$ in the plane, and for each $i \in [k]$, we fix a plane curve with endpoints $a_i$ and $o$ as the reference curve $\pi_i$. For an edge $e$, the parity of crossing with respect to $\pi_i$ defines the $i$-th bit of $\lab(e)$. The graph $G_\cs$ constructed in this fashion has $k$-bit labels and will be referred as $\emph{$k$-labeled graph}$.
\begin{definition}[labeled graphs]
For an integer $k \geq 1$, a \textbf{$k$-labeled graph} is a multi-graph $G = (V,E)$ and where each edge $e \in E$ has a label $\mathsf{lab}(e) \in \{0,1\}^k$ which is a $k$-bit binary string; we use $\mathsf{lab}_i(e)$ to denote the $i$-th bit of $\mathsf{lab}(e)$ for $i \in [k]$.
\end{definition}

A $P$-\textit{separator} refers to a subset $\cs' \subseteq \cs$ that separates all point-pairs $(a_i,a_j)$ for $(i,j) \in P$.
Our goal is to find a $P$-separator with the minimum size.
To this end, we first introduce the notion of \textit{labeled graphs} and some related concepts.

Let $G$ be a $k$-labeled graph.
For a cycle (or a path) $\gamma$ in $G$ with edge sequence $(e_1,\dots,e_r)$, we define $\mathsf{parity}(\gamma) = \bigoplus_{t=1}^r \mathsf{lab}(e_t)$ and denote by $\mathsf{parity}_i(\gamma)$ the $i$-th bit of $\mathsf{parity}(\gamma)$ for $i \in [k]$.
Here the notation ``$\oplus$'' denotes the bitwise XOR operation for binary strings.
Also, we define $\varPhi(\gamma)$ as the partition of $[k]$ consisting of two parts $I_0 = \{i: \mathsf{parity}_i(\gamma) = 0\}$ and $I_1 = \{i: \mathsf{parity}_i(\gamma) = 1\}$.
Next, we define an important notion called \textit{parity partition}.

\begin{definition}[parity partition]
Let $G$ be a $k$-labeled graph.
The \textbf{parity partition} induced by $G$, denoted by $\varPhi_G$, is the partition of $[k]$ defined as $\varPhi_G = \bigodot_{\gamma \in \varGamma_G} \varPhi(\gamma)$.
In other words, $i,j \in [k]$ belong to the same part of $\varPhi_G$ iff $\mathsf{parity}_i(\gamma) = \mathsf{parity}_j(\gamma)$ for every cycle $\gamma$ in $G$.
\end{definition}

The following two lemmas state some basic properties of the parity partition.

\begin{lemma} \label{lem-components}
Let $G$ be a $k$-labeled graph, and $C_1,\dots,C_r$ be the connected components of $G$ each of which is also regarded as a $k$-labeled graph.
Then $\varPhi_G = \bigodot_{t=1}^r \varPhi_{C_t}$.
\end{lemma}
\begin{proof}
Note that a cycle in $G$ must be contained in some connected component $C_t$ for $t \in [r]$, i.e., $\varGamma_G = \bigcup_{t=1}^r \varGamma_{C_t}$.
Thus, $\varPhi_G = \bigodot_{\gamma \in \varGamma_G} \varPhi(\gamma) = \bigodot_{t=1}^r (\bigodot_{\gamma \in \varGamma_{C_t}} \varPhi(\gamma)) = \bigodot_{t=1}^r \varPhi_{C_t}$.
\end{proof}

\begin{lemma} \label{lem-spanningtree}
Let $G$ be a connected $k$-labeled graph, and $T$ be a spanning tree of $G$.
Let $E_0$ be the edges of $G$ that are not in $T$.
Then $\varPhi_G = \bigodot_{e \in E_0} \varPhi(\gamma_e)$, where $\gamma_e$ is the cycle in $G$ consists of the edge $e$ and the (unique) simple path between the two endpoints of $e$ in $T$.
\end{lemma}
\begin{proof} The proof is similar to and more general form of Lemma~\ref{lem:tree-edge-walk}.
It is clear that $\varPhi_G \preceq \bigodot_{e \in E_0} \varPhi(\gamma_e)$ because $\gamma_e \in \varGamma_G$ for all $e \in E_0$.
To show $\varPhi_G \succeq \bigodot_{e \in E_0} \varPhi(\gamma_e)$, we use contradiction.
Assume $\varPhi_G \nsucceq \bigodot_{e \in E_0} \varPhi(\gamma_e)$.
Then there exist $i,j \in [k]$ which belong to different parts in $\varPhi_G$ but belong to the same part in $\bigodot_{e \in E_0} \varPhi(\gamma_e)$, i.e., $\mathsf{parity}_i(\gamma_e) = \mathsf{parity}_j(\gamma_e)$ for all $e \in E_0$.
Since $i$ and $j$ belong to different parts in $\varPhi_G$, we have $\mathsf{parity}_i(\gamma) \neq \mathsf{parity}_j(\gamma)$ for some $\gamma \in \varGamma_G$.
Let $\gamma^* \in \varGamma_G$ be the cycle satisfying $\mathsf{parity}_i(\gamma^*) \neq \mathsf{parity}_j(\gamma^*)$ that contains the smallest number of edges in $E_0$.
Note that $\gamma^*$ contains at least one edge in $E_0$, for otherwise $\gamma^*$ is a cycle in the tree $T$ and hence $\mathsf{parity}_i(\gamma^*) = \mathsf{parity}_j(\gamma^*) = 0$ (simply because a cycle in a tree goes through each edge even number of times).
Let $e = (u,v)$ be an edge of $\gamma^*$ that is in $E_0$.
We create a new cycle $\gamma'$ from $\gamma^*$ by replacing the edge $e$ in $\gamma^*$ with the (unique) simple path $\pi_{uv}$ between $u$ and $v$ in $T$.
Recall that $\mathsf{parity}_i(\gamma_e) = \mathsf{parity}_j(\gamma_e)$.
Since $\mathsf{parity}_i(\gamma_e) = \mathsf{lab}_i(e) \odot \mathsf{parity}_i(\pi_{uv})$ and $\mathsf{parity}_j(\gamma_e) = \mathsf{lab}_i(e) \odot \mathsf{parity}_j(\pi_{uv})$, we have $\mathsf{lab}_i(e) \odot \mathsf{parity}_i(\pi_{uv}) = \mathsf{lab}_j(e) \odot \mathsf{parity}_j(\pi_{uv})$.
Because $\mathsf{parity}_i(\gamma^*) \neq \mathsf{parity}_j(\gamma^*)$, we further have
\begin{equation*}
    \begin{aligned}
        \mathsf{parity}_i(\gamma') &= \mathsf{parity}_i(\gamma^*) \odot (\mathsf{lab}_i(e) \odot \mathsf{parity}_i(\pi_{uv})) \\
        &=\mathsf{parity}_i(\gamma^*) \odot (\mathsf{lab}_j(e) \odot \mathsf{parity}_j(\pi_{uv})) \\
        &\neq \mathsf{parity}_j(\gamma^*) \odot (\mathsf{lab}_j(e) \odot \mathsf{parity}_j(\pi_{uv})) = \mathsf{parity}_j(\gamma').
    \end{aligned}
\end{equation*}
However, this is impossible because $\gamma'$ has fewer edges in $E_0$ than $\gamma^*$ and $\gamma^*$ is the cycle satisfying $\mathsf{parity}_i(\gamma^*) \neq \mathsf{parity}_j(\gamma^*)$ that contains the smallest number of edges in $E_0$.
Therefore, $\varPhi_G \succeq \bigodot_{e \in E_0} \varPhi(\gamma_e)$ and hence $\varPhi_G = \bigodot_{e \in E_0} \varPhi(\gamma_e)$.
\end{proof}

Now we are ready to describe our algorithm.
The first step of our algorithm is to build a $k$-labeled graph $G_\mathcal{S}$ for the obstacle set $\mathcal{S}$.
The vertices of $G_\mathcal{S}$ are the obstacles in $\mathcal{S}$, and the labeled edges of $G_\mathcal{S}$ ``encode'' enough information for determining whether a subset of $\mathcal{S}$ is a $P$-separator.
Once we obtain $G_\mathcal{S}$, we can totally forget the input obstacles and points, and the rest of our algorithm will work on $G_\mathcal{S}$ only.

We build $G_\mathcal{S}$ as follows.
For each $S \in \cs$, we pick a reference point $\mathsf{ref}(S)$ inside the obstacle $S$.
Let $\mathsf{Arr}(\mathcal{S})$ denote the arrangement induced by the boundaries of the obstacles in $\mathcal{S}$, and $|\mathsf{Arr}(\mathcal{S})|$ be the complexity of $\mathsf{Arr}(\mathcal{S})$.
By assumption, $|\mathsf{Arr}(\mathcal{S})| = n^{O(1)}$.
We pick an arbitrary point $o$ in the plane, and for each $i \in [k]$, we fix a plane curve $\pi_i$ with endpoints $a_i$ and $o$.
We choose the curves $\pi_1,\dots,\pi_k$ carefully such that including them does not increase the complexity of the arrangement $\mathsf{Arr}(\mathcal{S})$ significantly.
Specifically, we require the complexity of the arrangement induced by the boundaries of the obstacles in $\mathcal{S}$ and these curves to be bounded by $k^{O(1)} \cdot |\mathsf{Arr}(\mathcal{S})|$, which is clearly possible.
As mentioned before, the vertices of $G_\mathcal{S}$ are the obstacles in $\mathcal{S}$.
The edge set $E_{G_\mathcal{S}}$ of $G_\mathcal{S}$ is defined as follows.
For each $i \in [k]$ and each $S \in \mathcal{S}$ such that $a_i \in S$, we include in $E_{G_\mathcal{S}}$ a self-loop $e$ on $S$ with $\mathsf{lab}_i(e) = 1$ and $\mathsf{lab}_{i'}(e) = 0$ for all $i' \in [k] \backslash \{i\}$.
For each pair $(S,S')$ of obstacles in $\cs$ and each $l \in \{0,1\}^k$, we include in $E_{G_\mathcal{S}}$ an edge $e = (S,S')$ with $\mathsf{lab}(e) = l$ if there exists a plane curve inside $S \cup S'$ with endpoints $\mathsf{ref}(S)$ and $\mathsf{ref}(S')$ which crosses $\pi_i$ an odd (resp., even) number of times for all $i \in [k]$ such that $a_i \notin S \cup S'$ and the $i$-th bit of $l$ is equal to 1 (resp., 0).
The next lemma shows $G_\mathcal{S}$ can be constructed in $2^{O(k)} n^{O(1)}$ time, as $|\mathsf{Arr}(\mathcal{S})| = n^{O(1)}$.

\begin{lemma}
\label{lem:graph-construction}
The $k$-labeled graph $G_\mathcal{S}$ can be constructed in $2^{O(k)} n^{O(1)} \cdot |\mathsf{Arr}(\mathcal{S})|$ time.
\end{lemma}
\begin{proof}
The self-loops of $G_\mathcal{S}$ can be constructed in $O(kn)$ time by checking for $i \in [k]$ and $S \in \mathcal{S}$ whether $a_i \in S$.
For each pair $(S,S')$ of obstacles in $\cs$, we show how to compute the edges in $G_\mathcal{S}$ between $S$ and $S'$ in $2^{O(k)} \cdot |\mathsf{Arr}(\mathcal{S})|$ time.
Let $K = \{i \in [k]: a_i \notin S \cup S'\}$; without loss of generality, assume $K = \{a_1,\dots,a_j\}$.
Denote by $\mathsf{Arr}(S,S')$ the arrangement induced by the boundary of $S \cup S'$ and the curves $\pi_1,\dots,\pi_j$, and define $\mathcal{F}$ as the set of faces of $\mathsf{Arr}(S,S')$ that are contained in $S \cup S'$.
See Figure~\ref{fig-arrange} for an illustration of the arrangement $\mathsf{Arr}(S,S')$.
We say two faces $F,F' \in \mathcal{F}$ are \textit{adjacent} if they share a common edge $\sigma(F,F')$ of $\mathsf{Arr}(S,S')$.
For two adjacent faces $F,F' \in \mathcal{F}$, we define $\theta(F,F') \in \{0,1\}^j$ by setting the $i$-th bit of $\theta(F,F')$ to be 1 for all $i \in [j]$ such that $\sigma(F,F')$ is a portion of $\pi_i$ and setting the other bits to be 0.
We construct a (unlabeled and undirected) graph $G$ with vertex set $\mathcal{F} \times \{0,1\}^j$ as follows.
For any two vertices $(F,l)$ and $(F',l')$ such that $F$ and $F'$ are adjacent and $l \oplus l' = \theta(F,F')$, we connect them by an edge in $G$.

\begin{figure}[htbp]
    \centering
    \includegraphics[height=5cm]{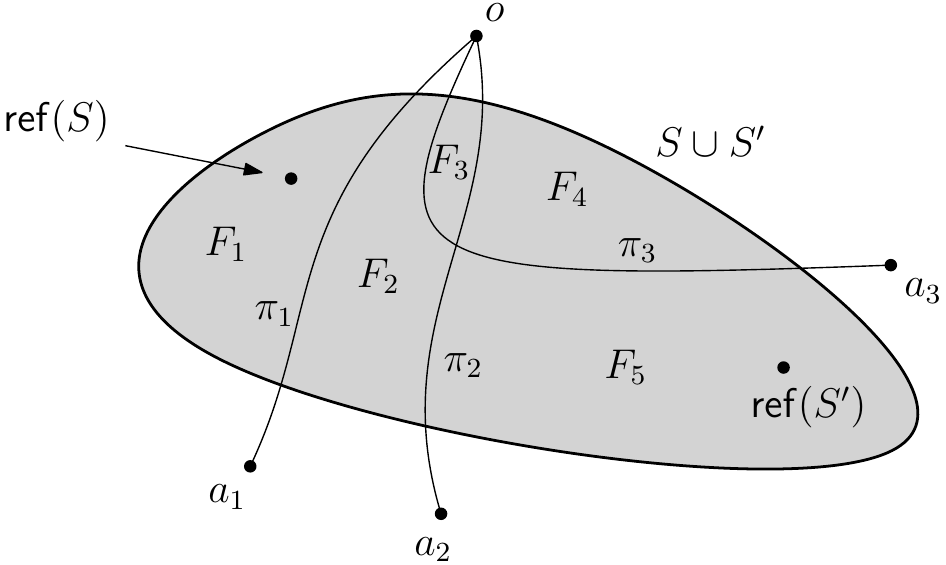}
    \caption{An illustration of the arrangement $\mathsf{Arr}(S,S')$. The grey area is $S \cup S'$. The set $\mathcal{F}$ consists of five faces $F_1,\dots,F_5$.}
    \label{fig-arrange}
\end{figure}

Let $F \in \mathcal{F}$ and $F' \in \mathcal{F}$ be the faces containing the reference points $\mathsf{ref}(S)$ and $\mathsf{ref}(S')$, respectively, and denote by $\mathbf{0} \in \{0,1\}^j$ the element with all bits 0.
We claim that there is an edge $(S,S')$ in $G_\cs$ with label $l$ iff the vertices $(F,\mathbf{0})$ and $(F',l)$ are in the same connected component of $G$.
To prove the claim, we first make a simple observation about the graph $G$ we constructed.
Let $(F_1,l_1)\dots,(F_m,l_m)$ be a path in $G$.
From the construction of $G$, it is easy to see (by a simple induction on $m$) that any plane curve from a point in $F_1$ to a point in $F_m$ that visits the faces $F_1,\dots,F_m$ in order crosses $\pi_i$ an odd (resp., even) number of times for all $i \in [j]$ such that the $i$-th bit of $l_1 \oplus l_m$ is equal to 1 (resp., 0).
Therefore, if there is a path in $G$ from $(F,\mathbf{0})$ to $(F',l)$, then there exists a plane curve from $\mathsf{ref}(S)$ to $\mathsf{ref}(S')$ that crosses $\pi_i$ an odd (resp., even) number of times for all $i \in [j]$ such that the $i$-th bit of $l$ is equal to 1 (resp., 0), which implies that there is an edge $(S,S')$ in $G_\cs$ with label $l$.
This proves the ``if'' part of the claim.
To see the ``only if'' part, assume there is an edge $(S,S')$ in $G_\cs$ with label $l$.
Then there exists a plane curve $\pi$ from $\mathsf{ref}(S)$ to $\mathsf{ref}(S')$ that crosses $\pi_i$ an odd (resp., even) number of times for all $i \in [j]$ such that the $i$-th bit of $l$ is equal to 1 (resp., 0).
Let $F_1,\dots,F_m$ be the sequence of faces visited by $\pi$ in order, where $F_1 = F$ and $F_m = F'$.
Then there is a path $(F_1,l_1),\dots,(F_m,l_m)$ in $G$ where $l_1 = \mathbf{0}$ and $l_t = l_{t-1} \odot \theta(F_{t-1},F_t)$ for $t \in [m] \backslash \{1\}$.
By our above observation, we have $l_1 \odot l_m = l$, which implies $l_m = l$.
It follows that $(F,\mathbf{0})$ and $(F',l)$ are in the same connected component of $G$.

By the above discussion, to compute the edges in $G_\mathcal{S}$ between $S$ and $S'$, it suffices to compute the connected component $C$ of $G$ that contains the vertex $(F,\mathbf{0})$: we have an edge $(S,S')$ in $G_\mathcal{S}$ with label $l \in \{0,1\}^k$ iff $(F',l') \in C$ where $l' \in \{0,1\}^j$ consists of the first $j$-bits of $l$.
The number of vertices and edges of $G$ is $2^O(k) \cdot |\mathsf{Arr}(\mathcal{S})|$, by our assumption that the complexity of the arrangement induced by the boundaries of the obstacles in $\mathcal{S}$ and the curves $\pi_1,\dots,\pi_k$ is bounded by $k^{O(1)} \cdot |\mathsf{Arr}(\mathcal{S})|$.
Therefore, $C$ can be computed in $2^O(k) \cdot |\mathsf{Arr}(\mathcal{S})|$ time.
As a result, $G_\mathcal{S}$ can be constructed in $2^{O(k)} n^{O(1)} \cdot |\mathsf{Arr}(\mathcal{S})|$ time.
\end{proof}

We say a $k$-labeled graph $G$ is \textit{$P$-good} if for all $(i,j) \in P$, $i$ and $j$ belong to different parts in $\varPhi_G$.
Note that if a subgraph of $G$ is $P$-good, then so is $G$.
The following key lemma establishes a characterization of $P$-separators using $P$-goodness. Note that the notion of $P$-goodness is almost the same as that
of well-behaved subgraphs from Lemma~\ref{lem:well-behaved}, except that it
is defined using parity partitions.

\begin{lemma} \label{lem-criterion}
A subset $\cs' \subseteq \cs$ is a $P$-separator iff the induced subgraph $G_\cs[\cs']$ is $P$-good.
\end{lemma}
\begin{proof}
We first introduce some notations.
For $(i,j) \in P$, denote by $\pi_{i,j}$ the plane curve with endpoints $a_i$ and $a_j$ obtained by concatenating the curves $\pi_i$ and $\pi_j$.
For each edge $e = (S,S')$ of $G_\cs$ with $S \neq S'$, we fix a \textit{representative curve} $\mathsf{rep}(e)$ of $e$, which is a plane curve contained in $S \cup S'$ with endpoints $\mathsf{ref}(S)$ and $\mathsf{ref}(S')$ that crosses $\pi_i$ an odd (resp., even) number of times for all $i \in [k]$ such that $\mathsf{lab}_i(e) = 1$ (resp., $\mathsf{lab}_i(e) = 0$); such a curve exists by our construction of $G_\cs$.

To prove the ``if'' part, assume $G_\cs[\cs']$ is $P$-good.
Let $(i,j) \in P$ be a pair and we want to show that $(a_i,a_j)$ is separated by $\cs'$.
If $a_i \in \bigcup_{S \in \cs'} S$ or $a_j \in \bigcup_{S \in \cs'} S$, we are done.
So assume $a_i \notin \bigcup_{S \in \cs'} S$ and $a_j \notin \bigcup_{S \in \cs'} S$.
Since $G_\cs[\cs']$ is $P$-good, there exists a cycle $\gamma$ in $G_\cs[\cs']$ such that $\mathsf{parity}_i(\gamma) \neq \mathsf{parity}_j(\gamma)$.
Without loss of generality, we assume $\mathsf{parity}_i(\gamma) = 0$ and $\mathsf{parity}_j(\gamma) = 1$.
Also, we can assume that $\gamma$ does not contain any self-loop edges; indeed, removing any self-loop edges from $\gamma$ does not change $\mathsf{parity}_i(\gamma)$ and $\mathsf{parity}_j(\gamma)$ because $a_i \notin \bigcup_{S \in \cs'} S$ and $a_j \notin \bigcup_{S \in \cs'} S$ (hence the $i$-th and $j$-th bits of the label of any self-loop on a vertex $S \in \mathcal{S}'$ are equal to 0).
Suppose the vertex sequence of $\gamma$ is $(S_0,\dots,S_r)$ where $S_0 = S_r$ and the edge sequence of $\gamma$ is $(e_1,\dots,e_r)$ where $e_t = (S_{t-1},S_t)$ for $t \in [r]$.
We concatenate the representative curves $\mathsf{rep}(e_1),\dots,\mathsf{rep}(e_r)$ to obtain a closed curve $\hat{\gamma}$ in the plane.
Because $\mathsf{parity}_i(\gamma) = 0$ and $\mathsf{parity}_j(\gamma) = 1$, $\pi_i$ crosses $\hat{\gamma}$ an even number of times and $\pi_j$ crosses $\hat{\gamma}$ an odd number of times.
It follows that $\pi_{i,j}$ crosses $\hat{\gamma}$ an odd number of times.
By the second statement of Fact~\ref{fact-separate}, $\hat{\gamma}$ separates $(a_i,a_j)$.
Since $\mathsf{rep}(e_t) \subseteq S_{t-1} \cup S_t$, we have $\hat{\gamma} \subseteq \bigcup_{t=1}^r S_t \subseteq \bigcup_{S \in \cs'} S$.
Therefore, $\cs'$ separates $(a_i,a_j)$.

\begin{figure}[htbp]
    \centering
    \includegraphics[height=5cm]{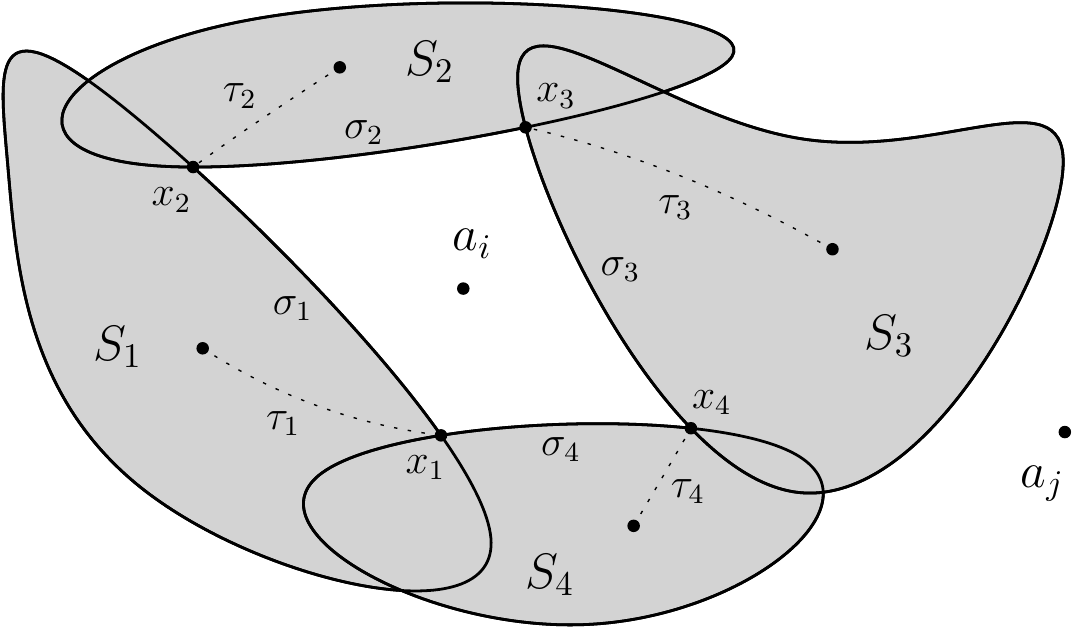}
    \caption{An illustration of the arcs $\sigma_1,\dots,\sigma_r$, the points $x_1,\dots,x_r$, and the curves $\tau_1,\dots,\tau_r$ (the points inside the obstacles are the reference points).}
    \label{fig-sepcycle}
\end{figure}

To prove the ``only if'' part, assume $\cs' \subseteq \cs$ is a $P$-separator, i.e., $\cs'$ separates all point-pairs $(a_i,a_j)$ for $(i,j) \in P$.
We want to show that $i$ and $j$ belong to different parts in $\varPhi_{G_\cs[\cs']}$ for all $(i,j) \in P$, or equivalently, for each $(i,j) \in P$ there exists a cycle $\gamma$ in $G_\cs[\cs']$ such that $\mathsf{parity}_i(\gamma) \neq \mathsf{parity}_j(\gamma)$.
Let $U = \bigcup_{S \in \cs'} S$.
We distinguish two cases: $\{a_i,a_j\} \cap U \neq \emptyset$ and $\{a_i,a_j\} \cap U = \emptyset$.
In the case $\{a_i,a_j\} \cap U \neq \emptyset$, we may assume $a_i \in U$ without loss of generality.
Then $a_i \in S$ for some $S \in \cs'$.
Therefore, by our construction of the graph $G_\cs$, there is a self-loop edge $e = (S,S)$ with $\mathsf{lab}_i(e) = 1$ and $\mathsf{lab}_{i'}(e) = 0$ for all $i' \in [k] \backslash \{i\}$.
The cycle $\gamma$ consists of this single edge is a cycle in $G_\cs[\cs']$ satisfying $\mathsf{parity}_i(\gamma) = 1 \neq 0 = \mathsf{parity}_j(\gamma)$.
Now it suffices to consider the case $\{a_i,a_j\} \cap U = \emptyset$.
The boundary $\partial U$ of $U$ consists of \textit{arcs} (each of which is a portion of the boundary of an obstacle in $\mathcal{S}'$) and \textit{break points} (each of which is an intersection point of the boundaries of two obstacles in $\mathcal{S}'$).
We can view $\partial U$ as a planar graph $G$ embedded in the plane, where the break points are vertices and the arcs are edges.
Each face of (the embedding of) $G$ is a connected component of $\mathbb{R}^2 \backslash \partial U$, which is either contained in $U$ (called \textit{in-faces}) or outside $U$ (called \textit{out-faces}).
Let $F_i$ and $F_j$ be the faces containing $a_i$ and $a_j$, respectively.
Since $\{a_i,a_j\} \cap U = \emptyset$, $F_i$ and $F_j$ are both out-faces.
Furthermore, we have $F_i \neq F_j$, for otherwise $a_i,a_j \in F_i$ and there exists a plane curve inside the out-face $F_i$ connecting $a_i$ and $a_j$, which contradicts the fact that $\cs'$ separates $(a_i,a_j)$.
Thus, there exists a simple cycle $\hat{\gamma}$ in $G$ (which corresponds to a simple closed curve in the plane) such that one of $F_i$ and $F_j$ is inside $\hat{\gamma}$ and the other one is outside $\hat{\gamma}$ (it is well-known that in a planar graph embedded in the plane, for any two distinct faces there exists a simple cycle in the graph such that one face is inside the cycle and the other is outside).
Because $a_i \in F_i$ and $a_j \in F_j$, we know that $\hat{\gamma}$ separates $(a_i,a_j)$ and hence $\pi_{i,j}$ crosses $\hat{\gamma}$ an odd number of times by the first statement of Fact~\ref{fact-separate}.
Let $\sigma_1,\dots,\sigma_r$ be the arcs of $\hat{\gamma}$ given in the order along $\hat{\gamma}$, and suppose they are contributed by the obstacles $S_1,\dots,S_r \in \cs'$, respectively (note that here $S_1,\dots,S_r$ need not be distinct).
For convenience, we write $\sigma_0 = \sigma_r$ and $S_0 = S_r$.
Let $x_t$ be the connection point of the arcs $\sigma_{t-1}$ and $\sigma_t$ for $t \in [r]$, then $x_t \in S_{t-1} \cap S_t$.
For each $t \in [r]$, we fix a plane curve $\tau_t$ inside the obstacle $S_t$ with endpoints $\mathsf{ref}(S_t)$ and $x_t$ (such a curve exists because $S_t$ is connected).
Again, we write $\tau_0 = \tau_r$.
See Figure~\ref{fig-sepcycle} for an illustration of the arcs $\sigma_1,\dots,\sigma_r$, the points $x_1,\dots,x_r$, and the curves $\tau_1,\dots,\tau_r$.
Now let $\tau_t'$ be the plane curve with endpoints $\mathsf{ref}(S_{t-1})$ and $\mathsf{ref}(S_t)$ obtained by concatenating $\tau_{t-1}$, $\sigma_{t-1}$, and $\tau_t$, and let $l_t \in \{0,1\}^k$ be the label whose $i'$-th bit is $0$ (resp., $1$) if $\pi_{i'}$ crosses $\tau_t'$ an even (resp., odd) number of times, for $t \in [r]$.
Note that $\tau_t' \subseteq S_{t-1} \cup S_t$.
Therefore, by our construction of $G_\cs$, there should be an edge $e_t = (S_{t-1},S_t)$ with $\mathsf{lab}(e) = l_t$, for each $t \in [r]$.
Consider the cycle $\gamma$ in $G_\cs[\cs']$ with vertex sequence $(S_0,\dots,S_r)$ and edge sequence $(e_1,\dots,e_t)$.
We claim that $\mathsf{parity}_i(\gamma) \neq \mathsf{parity}_j(\gamma)$.
Let $\gamma'$ be the closed plane curve obtained by concatenating the curves $\tau_1',\dots,\tau_r'$.
Observe that $\gamma'$ consists of $\hat{\gamma}$ and two copies of $\tau_1,\dots,\tau_r$.
It follows that for any plane curve $\pi$, the parity of the number of times that $\pi$ crosses $\gamma'$ is equal to the parity of the number of times that $\pi$ crosses $\hat{\gamma}$.
In particular, $\pi_{i,j}$ crosses $\gamma'$ an odd number of times.
Without loss of generality, we may assume that $\pi_i$ crosses $\gamma'$ an odd number of times and $\pi_j$ crosses $\gamma'$ an even number of times.
Since $\gamma'$ is the concatenation of $\tau_1',\dots,\tau_r'$ and the parity of the number of times that $\pi_i$ (resp., $\pi_j$) crosses $\tau_t'$ is indicated by the $i$-th (resp., $j$-th) bit of $l_t$, the $i$-th (resp., $j$-th) bit of $\bigodot_{t=1}^r l_t$ is $1$ (resp., $0$).
Because $\mathsf{parity}(\gamma) = \bigodot_{t=1}^r l_t$, we have $\mathsf{parity}_i(\gamma) \neq \mathsf{parity}_j(\gamma)$.
\end{proof}

\begin{definition}
Let $G = (V_G,E_G)$ and $H = (V_H,E_H)$ be two $k$-labeled graphs.
A \textbf{parity-preserving mapping (PPM)} from $H$ to $G$ is a pair $f = (f_V,f_E)$ consisting of two functions $f_V: V_H \rightarrow V_G$ and $f_E: E_H \rightarrow \varPi_G$ such that for each edge $e = (u,v) \in E_H$, $f_E(e)$ is a path between $f(u)$ and $f(v)$ in $G$ satisfying $\mathsf{parity}(f_E(e)) = \mathsf{lab}(e)$.
The \textbf{cost} of the PPM $f$ is defined as $\mathsf{cost}(f) = |V_H| - |E_H| + \sum_{e \in E_H} |f_E(e)|$.
The \textbf{image} of $f$, denoted by $\mathit{Im}(f)$, is the subgraph of $G$ consisting of the vertices $f_V(v)$ for $v \in V_H$ and the vertices on the paths $f_E(e)$ for $e \in E_H$, and the edges on the paths $f_E(e)$ for $e \in E_H$.
\end{definition}

\begin{fact} \label{fact-cost}
For any PPM $f$, the number of vertices of $\mathit{Im}(f)$ is at most $\mathsf{cost}(f)$.
\end{fact}
\begin{proof}
Let $f = (f_V,f_E)$ be a PPM from $H = (V_H,E_H)$ to $G$.
The number of vertices $f_V(v)$ for $v \in V_H$ is at most $|V_H|$.
The number of \textit{internal} vertices on each path $f_E(e)$ for $e \in E_H$ is at most $|f_E(e)| - 1$.
Note that a vertex of $\mathit{Im}(f)$ is either $f_V(v)$ for some $v \in V_H$ or an internal vertex on the path $f_E(e)$ for some $e \in E_H$.
Thus, the total number of vertices of $\mathit{Im}(f)$ is at most $|V_H| + \sum_{e \in E_H} (|f_E(e)| - 1) = |V_H| - |E_H| + \sum_{e \in E_H} |f_E(e)| = \mathsf{cost}(f)$.
\end{proof}

\begin{lemma} \label{lem-PPM1}
Let $H$ be a $P$-good $k$-labeled graph and $f$ be a PPM from $H$ to $G_\cs$.
Then $\mathit{Im}(f)$ is also $P$-good.
In particular, $\mathsf{cost}(f) \geq \mathsf{opt}$.
\end{lemma}
\begin{proof}
To see $\mathit{Im}(f)$ is $P$-good, what we want is that $i$ and $j$ belong to different parts of $\varPhi_{\mathit{Im}(f)}$ for all $(i,j) \in P$.
Consider a pair $(i,j) \in P$.
Since $H$ is $P$-good, there exists a cycle $\gamma$ in $H$ such that $\mathsf{parity}_i(\gamma) \neq \mathsf{parity}_j(\gamma)$.
Let $\gamma'$ be the image of $\gamma$ under $f$, which is a cycle in $\mathit{Im}(f)$ obtained by replacing each vertex $v$ of $\gamma$ with $f_V(v)$ and each edge $e$ of $\gamma$ with the path $f_E(e)$.
Because $f$ is a PPM, we have $\mathsf{parity}(\gamma') = \mathsf{parity}(\gamma)$.
Therefore, $\mathsf{parity}_i(\gamma') \neq \mathsf{parity}_j(\gamma')$.
It follows that $i$ and $j$ belong to different parts of $\varPhi_{\mathit{Im}(f)}$, and hence $\mathit{Im}(f)$ is $P$-good.
To see $\mathsf{cost}(f) \geq \mathsf{opt}$, let $\cs' \subseteq \cs$ be the vertex set $\mathit{Im}(f)$.
Then $\mathit{Im}(f)$ is a subgraph of $G_\cs[\cs']$, which implies $G_\cs[\cs']$ is also $P$-good.
By Lemma~\ref{lem-criterion}, $\cs'$ is a $P$-separator, i.e., $|\cs'| \geq \mathsf{opt}$.
Furthermore, by Fact~\ref{fact-cost}, we have $\mathsf{cost}(f) \geq |\cs'| \geq \mathsf{opt}$.
\end{proof}

\begin{lemma} \label{lem-PPM2}
There exists a $P$-good $k$-labeled graph $H^*$ with at most $4k$ vertices and $5k$ edges and a PPM $f^*$ from $H^*$ to $G_\cs$ such that $\mathsf{cost}(f^*) = \mathsf{opt}$.
\end{lemma}
\begin{proof}
Let $\cs_\text{opt} \subseteq \cs$ be a $P$-separator of the minimum size.
By Lemma~\ref{lem-criterion}, the induced subgraph $G_\cs[\cs_\text{opt}]$ is $P$-good.
Let $G$ be a \textit{minimal} $P$-good subgraph of $G_\cs[\cs_\text{opt}]$, that is, no proper subgraph of $G$ is $P$-good.
Note that $G$ does not have degree-0 and degree 1 vertices, simply because deleting a degree-0 or degree-1 vertex (and its adjacent edge) from $G$ does not change $\varPhi_G$.
Suppose $G$ has $r$ connected components $C_1,\dots,C_r$.
We fix a spanning tree $T_t$ of $C_t$ for each $t \in [r]$.
Let $E_0$ be the set of non-tree edges of $G$, i.e., the edges not in $T_1,\dots,T_r$.
We mark all vertices of $G$ with degree at least 3.
Furthermore, for each component $C_t$ that has no vertex with degree at least 3 (which should be a simple cycle because $G$ does not have degree-1 vertices), we mark a vertex of $C_t$ that is adjacent to the (only) non-tree edge of $C_t$.
We notice that all unmarked vertices of $G$ are of degree 2 and each component $C_t$ of $G$ has at least one marked vertex.
Therefore, $G$ consists of the marked vertices and a set $K$ of \textit{chains} (i.e., paths consisting of degree-2 vertices) connecting marked vertices.
See (the left and middle figures of) Figure~\ref{fig-contraction} for an illustration of the marked vertices and chains.

We claim that $|E_0| < k$, the number of marked vertices in $G$ is bounded by $4k$, and $|K| \leq 5k$.
For each $e = (u,v) \in E_0$, let $\gamma_e$ be the (simple) cycle consists of $e$ and the (unique) simple path between $u$ and $v$ in $T_t$, where $t \in [r]$ is the index such that $C_t$ contains $u$ and $v$.
By Lemma~\ref{lem-components} and~\ref{lem-spanningtree}, we have $\varPhi_{G} = \bigodot_{t=1}^r \varPhi_{C_t} = \bigodot_{e \in E_0} \varPhi(\gamma_e)$.
By Fact~\ref{fact-represent}, there exists $E_0' \subseteq E_0$ with $|E_0'| < k$ such that $\bigodot_{e \in E_0'} \varPhi(\gamma_e) = \bigodot_{e \in E_0} \varPhi(\gamma_e)$.
Let $G'$ be the subgraph of $G$ obtained by removing all edges in $E_0 \backslash E_0'$.
Using Lemma~\ref{lem-components} and~\ref{lem-spanningtree} again, we deduce that 
\begin{equation} \label{eq-forest}
    \varPhi_{G'} = \bigodot_{e \in E_0'} \varPhi(\gamma_e) = \bigodot_{e \in E_0} \varPhi(\gamma_e) = \varPhi_G.
\end{equation}
Therefore, $G'$ is also $P$-good.
It follows that $G' = G$, since no proper subgraph of $G$ is $P$-good.
This further implies $E_0' = E_0$ and $|E_0| < k$.
Next, we consider the number of vertices in $G$ with degree at least 3.
Since $G$ does not have degree-1 vertices, any leaf of the trees $T_1,\dots,T_r$ must be adjacent to some edge in $E_0$.
Since $|E_0| < k$, the number of leaves of $T_1,\dots,T_r$ is at most $2k$, and hence there are at most $2k$ nodes in $T_1,\dots,T_r$ whose degree is at least 3.
Now observe that a marked vertex $v$ of $G$ is either adjacent to some edge in $E_0$ or of degree at least 3 in the tree $T_t$, where $C_t$ is the component containing $v$.
Therefore, there can be at most $4k$ marked vertices in $G$.
Finally, we bound $|K|$, the number of chains.
Note that each edge of $G$ belongs to exactly one chain in $K$.
Therefore, the number of chains containing at least one edge in $E_0$ is at most $k$, because $|E_0| < k$.
All the other chains, i.e., the chains that do not have any edge in $E_0$, are contained in the trees $T_1,\dots,T_r$.
It follows that these chains do not form any cycle, and thus their number is less than the number of marked vertices in $G$ (which is at most $4k$).
Thus, $G$ has at most $5k$ chains, i.e., $|K| \leq 5k$.

\begin{figure}[htbp]
    \centering
    \includegraphics[height=5.5cm]{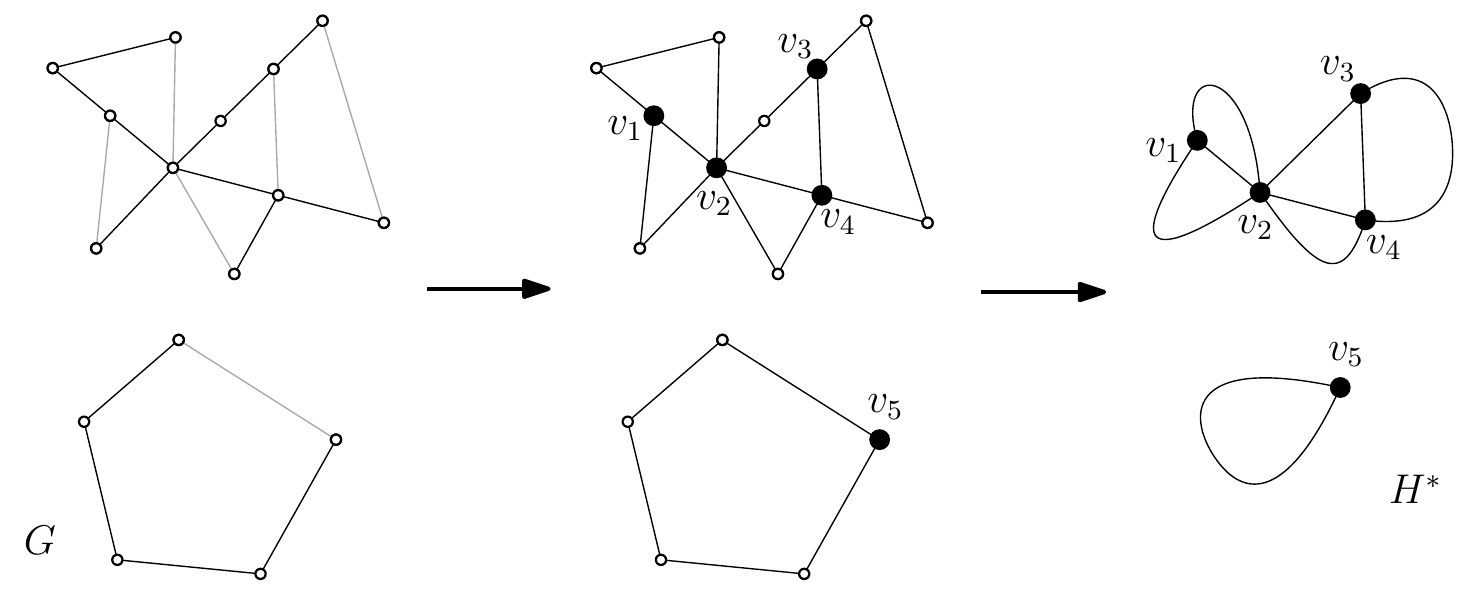}
    \caption{An illustration of the marked vertices in $G$ and the resulting graph $H^*$ by path-contraction. The left figure shows the graph $G$ consisting of two connected components where the black edges are tree edges and the grey edges are non-tree edges in $E_0$. The middle figure shows the marked vertices in $G$ (and the chains in $K$ connecting the marked vertices). The right figure shows the graph $H^*$ obtained by path-contraction.}
    \label{fig-contraction}
\end{figure}

The desired $k$-labeled graph $H^*$ is defined via a path-contraction procedure on $G$ as follows.
The vertices of $H^*$ are one-to-one corresponding to the marked vertices of $G$.
The edges of $H^*$ are one-to-one corresponding to the chains in $K$: for each chain connecting two marked vertices $u$ and $v$, we have an edge in $H^*$ connecting the two vertices of $H$ corresponding to $u$ and $v$.
The label of each edge $e$ of $H^*$ is defined as $\mathsf{lab}(e) = \mathsf{parity}(\pi_e)$, where $\pi_e$ is the chain in $C$ corresponding to $e$.
See Figure~\ref{fig-contraction} for an illustration of how to obtain $H^*$ via path-contraction.
Since there are at most $4k$ marked vertices in $G$ and $|K| \leq 5k$, $H^*$ has at most $4k$ vertices and $5k$ edges.
Next, we define the PPM $f^* = (f_V^*,f_E^*)$ from $H^*$ to $G_\cs$.
The function $f_V^*$ simply maps each vertex of $H^*$ to its corresponding marked vertex in $G$ (which is a vertex of $G_\cs$), and the function $f_E^*$ simply maps each edge of $H^*$ to its corresponding chain in $K$ (which is a path in $G_\cs$).
The fact that $f^*$ is a PPM directly follows from the construction of $H^*$.
Furthermore, we observe that $\mathsf{cost}(f^*)$ is equal to the number of vertices in $G$, because the chains in $K$ are ``interior-disjoint'' in the sense that two chains can only intersect at their endpoints.
Therefore, $\mathsf{cost}(f^*) = |\cs_\text{opt}| = \mathsf{opt}$.
Finally, we show that $H^*$ is $P$-good.
It suffices to show $\varPhi_{H^*} = \varPhi_G$.
Consider two elements $i,j \in [k]$ belong to the same part of $\varPhi_G$.
We have $\mathsf{parity}_i(\gamma) = \mathsf{parity}_j(\gamma)$ for any cycle $\gamma$ in $G$.
It follows that $\mathsf{parity}_i(\gamma^*) = \mathsf{parity}_j(\gamma^*)$ for any cycle $\gamma^*$ in $H^*$, because the image of $\gamma^*$ under $f^*$ is a cycle $\gamma$ in $G$ satisfying $\mathsf{parity}(\gamma) = \mathsf{parity}(\gamma^*)$.
Thus, $i$ and $j$ belong to the same part of $\varPhi_{H^*}$.
Next consider two elements $i,j \in [k]$ belong to different parts of $\varPhi_G$.
By Equation~\ref{eq-forest}, there exists some edge $e \in E_0$ such that $i$ and $j$ belong to different parts of $\varPhi(\gamma_e)$, i.e., $\mathsf{parity}_i(\gamma_e) \neq \mathsf{parity}_j(\gamma_e)$.
Since $\gamma_e$ is a simple cycle in $G$, it corresponds to a simple cycle in $H^*$, i.e., there is a simple cycle $\gamma^*$ in $H^*$ whose image under $f^*$ is $\gamma_e$.
Because $f^*$ is a PPM, we have $\mathsf{parity}(\gamma^*) = \mathsf{parity}(\gamma_e)$.
It then follows that $\mathsf{parity}_i(\gamma^*) \neq \mathsf{parity}_j(\gamma^*)$ and hence $i,j$ belong to different parts of $\varPhi_{H^*}$.
Therefore, $\varPhi_{H^*} = \varPhi_G$ and $H^*$ is $P$-good.
\end{proof}

The above lemma already gives us an algorithm that runs in $2^{O(k^2)} n^{O(k)}$ time.
First, we guess the $k$-labeled graph $H^*$ in Lemma~\ref{lem-PPM2}.
Since $H^*$ has at most $4k$ vertices and $5k$ edges, the number of possible graph structures of $H^*$ is $k^{O(k)}$ and the number of possible labeling of the edges of $H^*$ is bounded by $(2^k)^{5k}$.
Therefore, there can be $2^{O(k^2)}$ possibilities for $H^*$.
We enumerate all possible $H^*$, and for every $H^*$ that is $P$-good, we compute a PPM from $H^*$ to $G_\cs$ with the minimum cost; later we will show how to do this in $n^{O(k)}$ time.
Among all these PPMs, we take the one with the minimum cost, say $f^*$.
By Lemma~\ref{lem-PPM1} and~\ref{lem-PPM2}, we know that $\mathit{Im}(f^*)$ is $P$-good and $\mathsf{cost}(f^*) = \mathsf{opt}$.
To find an optimal solution, let $\cs' \subseteq \cs$ be the set of vertices of $\mathit{Im}(f^*)$.
Since $\mathit{Im}(f^*)$ is a subgraph of $G_\cs[\cs']$ and $\mathit{Im}(f^*)$ is $P$-good, we know that $G_\cs[\cs']$ is also $P$-good and hence $\cs'$ is a $P$-separator.
Furthermore, Fact~\ref{fact-cost} implies that $|\cs'| \leq \mathsf{cost}(f^*) = \mathsf{opt}$.
Therefore, $\cs'$ is an optimal solution for the problem instance.
The entire algorithm takes $2^{O(k^2)} n^{O(k)}$ time.

Now we discuss the missing piece of the above algorithm, how to compute a PPM from $H^*$ to $G_\cs$ with the minimum cost in $n^{O(k)}$ time, given a $k$-labeled graph $H^*=(V_{H^*},E_{H^*})$ with at most $4k$ vertices and $5k$ edges.
For all $u,v \in \cs$ and $l \in \{0,1\}^k$, let $\pi_{u,v,l}$ be the shortest path (i.e., the path with fewest edges) between $u$ and $v$ whose parity is $l$.
All these paths can be computed in $2^{O(k)} n^3$ time using Floyd's algorithm.
Suppose $f^* = (f_V^*,f_E^*)$ is the PPM from $H^*$ to $G_\cs$ we want to compute.
Recall that $\mathsf{cost}(f^*) = |V_{H^*}| - |E_{H^*}| + \sum_{e^* \in E_{H^*}} |f_E^*(e^*)|$.
The terms $|V_{H^*}|$ and $|E_{H^*}|$ only depend on $H^*$ itself.
Therefore, we want to choose $f^*$ that minimizes $\sum_{e^* \in E_{H^*}} |f_E^*(e^*)|$.
We simply enumerate all possibilities of $f_V^*$.
Since $H^*$ has at most $4k$ vertices, there are at most $n^{4k}$ possible $f_V^*$ to be considered.
Once $f_V^*$ is determined, the endpoints of the paths $f_E^*(e^*)$ are also determined.
This allows us to minimize $|f_E^*(e^*)|$ for each $e^* \in E_{H^*}$ independently.
Let $e^* = (u^*,v^*) \in E_{H^*}$.
Since $f^*$ is a PPM, $f_E^*(e^*)$ must be a path connecting $u = f_V^*(u)$ and $v = f_V^*(v)$ whose parity is $l = \mathsf{lab}(e^*)$.
By the definition of $\pi_{u,v,l}$, it follows that $|f_E^*(e^*)| \geq |\pi_{u,v,l}|$ and thus setting $f_E^*(e^*) = \pi_{u,v,l}$ will minimize $|f_E^*(e^*)|$.
After trying all possible $f_V^*$, we can finally find the optimal PPM $f^*$ in $n^{O(k)}$ time.

\subsection{\boldmath Improving the running time to $2^{O(p)} n^{O(k)}$} \label{sec-2^p}
To further improve the running time of the above algorithm to $2^{O(p)} n^{O(k)}$ requires nontrivial efforts.
Without loss of generality, in this section, we assume $k \leq n$.
Indeed, if $k > n$, the problem can be solved in $2^{O(k)}$ time by enumerating every subset $\cs' \subseteq \cs$ and checking if $\cs'$ is a $P$-separator (which can be done in polynomial time by first computing $\varPhi_{\cs'}$ using Lemma~\ref{lem-components} and~\ref{lem-spanningtree} and then applying the criterion of Lemma~\ref{lem-criterion}).

As stated before, there are $2^{O(k^2)}$ possibilities for $H^*$.
Thus, in order to improve the factor $2^{O(k^2)}$ to $2^{O(p)}$, we have to avoid enumerating all possible $H^*$.
Instead, we only enumerate the graph structure of $H^*$ (but not the labels of its edges).
There are $k^{O(k)}$ possible graph structures to be considered, because $H^*$ has at most $4k$ vertices and $5k$ edges.
For each possible graph structure, we want to label the edges to make $H^*$ $P$-good and then find a PPM from $H^*$ (with that labeling) to $G_\cs$ such that the cost of the PPM is minimized.
Formally, consider a graph structure $H^* = (V_{H^*},E_{H^*})$ of $H^*$.
A \textit{labeling-PPM pair} for $H^*$ refers to a pair $(\mathsf{lab},f^*)$ where $\mathsf{lab}: E_{H^*} \rightarrow \{0,1\}^k$ is a labeling for $H^*$ and $f^* = (f_V^*,f_E^*)$ is a PPM from $H^*$ to $G_\cs$ (with respect to the labeling $\mathsf{lab}$).
Our task is to find a labeling-PPM pair $(\mathsf{lab},f^*)$ for $H^*$ with the minimum $\mathsf{cost}(f^*)$ such that $H^*$ is $P$-good with respect to the labeling $\mathsf{lab}$.

Let $C_1,\dots,C_r$ be the connected components of $H^*$, and $T_1,\dots,T_r$ be spanning trees of $C_1,\dots,C_r$, respectively.
Let $E_0 \subseteq E_{H^*}$ be the set of edges that are not in $T_1,\dots,T_r$.
For each $e \in E_0$, denote by $\gamma_e$ the cycle in $H^*$ consisting of the edge $e$ and the (unique) simple path between the two endpoints of $e$ in $T_t$, where $t \in [r]$ is the index such that $C_t$ contains $e$.
By Lemma~\ref{lem-components} and~\ref{lem-spanningtree}, we have $\varPhi_{H^*} = \bigodot_{t=1}^r \varPhi_{C_t} = \bigodot_{e \in E_0} \varPhi(\gamma_e)$.
Therefore, a labeling makes $H^*$ $P$-good iff for every $(i,j) \in P$ there exists an edge $e \in E_0$ such that $\mathsf{parity}_i(\gamma_e) \neq \mathsf{parity}_j(\gamma_e)$ with respect to that labeling.
We say a labeling $\mathsf{lab}: E_{H^*} \rightarrow \{0,1\}^k$ \textit{respects} a function $\xi:P \rightarrow E_0$ if for all $(i,j) \in P$, we have $\mathsf{parity}_i(\gamma_e) \neq \mathsf{parity}_j(\gamma_e)$ where $e = \xi(i,j)$ and $\mathsf{parity}$ is calculated with respect to the labeling $\mathsf{lab}$.
Then we immediately have the following fact.
\begin{fact} \label{fact-respect}
A labeling makes $H^*$ $P$-good iff it respects some function $\xi:P \rightarrow E_0$.
\end{fact}

Our first observation is that for any function $\xi:P \rightarrow E_0$, one can efficiently find the ``optimal'' labeling-PPM pair $(\mathsf{lab},f^*)$ for $H^*$ satisfying the condition that $\mathsf{lab}$ respects $\xi$.
\begin{lemma} \label{lem-respect}
Given $\xi:P \rightarrow E_0$, one can compute in $2^{O(p)} n^{O(k)}$ time a labeling-PPM pair $(\mathsf{lab},f^*)$ for $H^*$ which minimizes $\mathsf{cost}(f^*)$ subject to the condition that $\mathsf{lab}$ respects $\xi$.
\end{lemma}
\begin{proof}
Suppose $f^* = (f_V^*,f_E^*)$ is the PPM we want to compute.
We enumerate all possibilities of $f_V^*:V_{H^*} \rightarrow \mathcal{S}$.
Since $|V_{H^*}| \leq 4k$, there are $n^{O(k)}$ different $f_V^*$ to be considered.
Fixing a function $f_V^*$, we want to determine the labeling $\mathsf{lab}$ and the function $f_E^*$ such that \textbf{(i)} $\mathsf{lab}$ respects $\xi$, \textbf{(ii)} $f^*$ is a PPM with respect to the labeling $\mathsf{lab}$, and \textbf{(iii)} $\mathsf{cost}(f^*)$ is minimized.
For an edge $e^* = (u^*,v^*) \in E_{H^*}$ and a label $l \in \{0,1\}^k$, we denote by $\mathsf{len}(e^*,l) = |\pi_{u,v,l}|$, where $u = f_V^*(u^*)$, $v = f_V^*(v^*)$.
As argued before, for a fixed labeling $\mathsf{lab}$, an optimal function $f_E^*$ is the one that maps each edge $e^* = (u^*,v^*) \in E_{H^*}$ to the path $\pi_{u,v,l}$, where $u = f_V^*(u^*)$, $v = f_V^*(v^*)$, $l = \mathsf{lab}(e^*)$; with this choice of $f_E^*$, we have $\mathsf{cost}(f^*) = |V_{H^*}| - |E_{H^*}| + \sum_{e^* \in E_{H^*}} \mathsf{len}(e^*,\mathsf{lab}(e))$.
Therefore, our actual task is to find a labeling $\mathsf{lab}$ that respects $\xi$ and minimizes $\sum_{e^* \in E_{H^*}} \mathsf{len}(e^*,\mathsf{lab}(e^*))$.
Suppose $E_{H^*} = \{e_1,\dots,e_m\}$ where $m = O(k)$.
Let $\delta: [m] \times E_0 \rightarrow \{0,1\}$ be an indicator defined as $\delta(t,e) = 1$ if $e_t$ is an edge of the cycle $\gamma_e$ and $\delta(t,e) = 0$ otherwise.
For a labeling $\mathsf{lab}: E_{H^*} \rightarrow \{0,1\}^k$, we have $\mathsf{parity}(\gamma_e) = \sum_{t=1}^m \delta(t,e) \cdot \mathsf{lab}(e_t)$ for any $e \in E_0$.
Therefore, a labeling $\mathsf{lab}$ respects $\xi$ iff $\sum_{t=1}^m \delta(t,\xi(i,j)) \cdot \mathsf{lab}_i(e_t) \neq \sum_{t=1}^m \delta(t,\xi(i,j)) \cdot \mathsf{lab}_j(e_t)$ for all $(i,j) \in P$, or equivalently, $\sum_{t=1}^m \delta(t,\xi(i,j)) \cdot (\mathsf{lab}_i(e_t) \oplus \mathsf{lab}_j(e_t)) = 1$ for all $(i,j) \in P$.
So our task is to find a labeling $\mathsf{lab}$ which minimizes $\sum_{t=1}^m \mathsf{len}(e_t,\mathsf{lab}(e_t))$ subject to $\sum_{t=1}^m \delta(t,\xi(i,j)) \cdot (\mathsf{lab}_i(e_t) \oplus \mathsf{lab}_j(e_t)) = 1$ for all $(i,j) \in P$.

Now consider the following problem: for a pair $(t',\phi)$ where $t' \in [m]$ is an index and $\phi:P \rightarrow \{0,1\}$ is a function, compute a ``partial'' labeling $\mathsf{lab}:\{e_1,\dots,e_{t'}\} \rightarrow \{0,1\}^k$ such that $\sum_{t=1}^{t'} \mathsf{len}(e_t,\mathsf{lab}(e_t))$ is minimized subject to the condition $\sum_{t=1}^{t'} \delta(t,\xi(i,j)) \cdot (\mathsf{lab}_i(e_t) \oplus \mathsf{lab}_j(e_t)) = \phi(i,j)$ for all $(i,j) \in P$.
We want to solve the problem for all pairs $(t',\phi)$.
This can be achieved using dynamic programming as follows.
For a label $l \in \{0,1\}^k$, we denote by $\phi_l:P \rightarrow \{0,1\}$ the function which maps $(i,j) \in P$ to 0 (resp., 1) if the $i$-th bit and the $j$-th bit of $l$ is the same (resp., different).
We consider the index $t'$ from $1$ to $m$.
Suppose now the problems for all pairs with index $t'-1$ have been solved.
To solve for a pair $(t',\phi)$, we enumerate the labeling $\mathsf{lab}(e_{t'})$ for $e_{t'}$.
Fixing $\mathsf{lab}(e_{t'}) = l$, the remaining problem becomes to determine $\mathsf{lab}:\{e_1,\dots,e_{t'-1}\} \rightarrow \{0,1\}^k$ that minimizes $\sum_{t=1}^{t'-1} \mathsf{len}(e_t,\mathsf{lab}(e_t))$ subject to the condition $\sum_{t=1}^{t'-1} \delta(t,\xi(i,j)) \cdot (\mathsf{lab}_i(e_t) \oplus \mathsf{lab}_j(e_t)) = \phi(i,j) \odot \phi_l(i,j)$ for all $(i,j) \in P$, which is exactly the problem for the pair $(t'-1,\phi \odot \phi_l)$.
Thus, provided that we already know the solution for the problem for all pairs with index $t'-1$, we can solve the problem for $(t',\phi)$ in $2^p \cdot p^{O(1)}$ time.
Since there are $2^p m$ pairs $(t',\phi)$ to be considered and $m = O(k)$, the problem for all pairs can be solved in $2^{O(p)}$ time.

Now we see that for a fixed $f_V^*$, one can compute in $2^{O(p)}$ time the optimal $\mathsf{lab}$ and $f_E^*$.
Since there are $n^{O(k)}$ possible $f_V^*$ to be considered, the entire algorithm takes $2^{O(p)} n^{O(k)}$ time, which completes the proof.
\end{proof}

The above lemma directly gives us a $k^{O(p)} n^{O(k)}$-time algorithm to compute the desired labeling-PPM pair.
By Fact~\ref{fact-respect}, it suffices to compute a labeling-PPM pair $(\mathsf{lab},f^*)$ for $H^*$ with the minimum $\mathsf{cost}(f^*)$ such that $\mathsf{lab}$ respects some function $\xi:P \rightarrow E_0$.
Note that the number of different functions $\xi: P \rightarrow E_0$ is at most $(5k)^p$ because $|P| = p$ and $|E_0| \leq 5k$.
We simply enumerate all these functions, and for each function $\xi: P \rightarrow E_0$, we use Lemma~\ref{lem-respect} to compute in $2^{O(p)} n^{O(k)}$ time a labeling-PPM pair $(\mathsf{lab},f^*)$ for $H^*$ with the minimum $\mathsf{cost}(f^*)$ such that $\mathsf{lab}$ respects $\xi$.
Among all the labeling-PPM pairs are computed, we then pick the pair $(\mathsf{lab},f^*)$ with the minimum $\mathsf{cost}(f^*)$.

To compute the desired labeling-PPM pair more efficiently, we observe that in fact, we do not need to try all functions $\xi: P \rightarrow E_0$.
If a family $\varXi$ of functions $\xi:P \rightarrow E_0$ satisfies that any labeling making $H^*$ $P$-good respects some $\xi \in \varXi$, then trying the functions in $\varXi$ is already sufficient.
We show the existence of such a family $\varXi$ of size $k^{O(k)}$.

\begin{lemma} \label{lem-representative}
There exists a family $\varXi$ of $k^{O(k)}$ functions $\xi:P \rightarrow E_0$ such that any labeling making $H^*$ $P$-good respects some $\xi \in \varXi$.
Furthermore, $\varXi$ can be computed in $k^{O(k)}$ time.
\end{lemma}
\begin{proof}
As the first step of our proof, we establish a bound on the number of sequences of ``finer and finer'' partitions of $[k]$.
Let $m \geq 1$ be an integer.
An $m$-sequence $(\varPhi_1,\dots,\varPhi_m)$ of partitions of $[k]$ is \textit{finer and finer} if $\varPhi_1 \succeq \dots \succeq \varPhi_m$.
We show that the total number of finer and finer $m$-sequences is bounded by $(m+k)^{O(k)}$.
To this end, we first observe that the number of non-decreasing sequences $(z_1,\dots,z_m)$ of integers in $[k]$ is $\binom{m+k-1}{k-1} = (m+k)^{O(k)}$.
Therefore, it suffices to show that for any non-decreasing sequence $(z_1,\dots,z_m)$ of integers in $[k]$, the number of finer and finer $m$-sequences $(\varPhi_1,\dots,\varPhi_m)$ satisfying $|\varPhi_i| = z_i$ for all $i \in [m]$ is bounded by $(m+k)^{O(k)}$.
Fix a non-decreasing sequence $(z_1,\dots,z_m)$ of integers in $[k]$.
For convenience, define $\varPhi_{m+1} = \{\{1\},\dots,\{k\}\}$ as finest partition of $[k]$ and let $z_{m+1} = |\varPhi_{m+1}| = k$.
Then we must have $\varPhi_m \succeq \varPhi_{m+1}$.
By applying Fact~\ref{fact-coarse}, for a fixed $\varPhi_{i+1}$ with $|\varPhi_{i+1}| = z_{i+1}$, the number of partitions $\varPhi_i \succeq \varPhi_{i+1}$ with $|\varPhi_i| = z_i$ is $z_{i+1}^{O(d_{i+1})}$ where $d_{i+1} = z_{i+1} - z_i$.
Therefore, by a simple induction argument we see that for an index $t \in [m]$, the number of the possibilities of the subsequence $(\varPhi_t,\dots,\varPhi_m)$ is bounded by $\prod_{i=t}^m z_{i+1}^{O(d_{i+1})} = k^{O(k-z_t)}$.
In particular, the number of finer and finer $m$-sequences $(\varPhi_1,\dots,\varPhi_m)$ satisfying $|\varPhi_i| = z_i$ for all $i \in [m]$ is bounded by $k^{O(k)}$.
Furthermore, we observe that these sequences can be computed in $O(m) + k^{O(k)}$ time by repeatedly using Fact~\ref{fact-coarse}.
Indeed, by Fact~\ref{fact-coarse}, for a fixed subsequence $(\varPhi_{t+1},\dots,\varPhi_m)$, one can compute in $k^{O(d_{t+1})}$ time all $\varPhi_t$ such that $|\varPhi_t| = z_t$ and $\varPhi_t \succeq \varPhi_{t+1}$ time, where $d_{t+1} = z_{t+1} - z_t$.
Therefore, knowing all $k^{O(k-z_{t+1})}$ possible subsequences $(\varPhi_{t+1},\dots,\varPhi_m)$, one can compute all possible subsequences $(\varPhi_t,\dots,\varPhi_m)$ in $k^{O(k-z_t)}$ time.
In particular, all finer and finer $m$-sequences $(\varPhi_1,\dots,\varPhi_m)$ satisfying $|\varPhi_i| = z_i$ for all $i \in [m]$ can be computed in $O(m) + k^{O(k)}$ time.
The $(m+k)^{O(k)}$ non-decreasing sequences $(z_1,\dots,z_m)$ of integers in $[k]$ can be easily enumerated in $(m+k)^{O(k)}$ time, which implies that all finer and finer $m$-sequences of partitions of $[k]$ can be computed in $(m+k)^{O(k)}$ time.

With the above result, we are now ready to prove the lemma.
Suppose $E_0 = \{e_1,\dots,e_m\}$ where $m = O(k)$.
We construct a family $\varXi$ of functions $\xi:P \rightarrow E_0$ as follows.
For every finer and finer $m$-sequence $(\varPhi_1,\dots,\varPhi_m)$ of partitions of $[k]$ satisfying that $i$ and $j$ belong to different parts in $\varPhi_m$ for all $(i,j) \in P$, we include in $\varXi$ a corresponding function $\xi:P \rightarrow E_0$ defined by setting $\xi(i,j) = e_t$ where $t \in [m]$ is the smallest index such that $i$ and $j$ belong to different parts in $\varPhi_t$.
By the above result, we have $|\varXi| = k^{O(k)}$ and $\varXi$ can be computed in $k^{O(k)}$ time.
It suffices to prove that $\varXi$ satisfies the desired property.
Let $\mathsf{lab}:E_{H^*} \rightarrow \{0,1\}^k$ be a labeling that makes $H^*$ $P$-good.
Recall that we have $\varPhi_{H^*} = \bigodot_{t=1}^m \varPhi(\gamma_{e_t})$.
Now we define a finer and finer $m$-sequence $(\varPhi_1,\dots,\varPhi_m)$ of partitions of $[k]$ by setting $\varPhi_t = \bigodot_{s=1}^t \varPhi(\gamma_{e_s})$ for all $t \in [m]$.
Then we have $\varPhi_m = \varPhi_{H^*}$.
Since $H^*$ is $P$-good, we know that $i$ and $j$ belong to different parts in $\varPhi_m$ for all $(i,j) \in P$.
Let $\xi \in \varXi$ be the function corresponding to the sequence $(\varPhi_1,\dots,\varPhi_m)$.
We shall show that $\mathsf{lab}$ respects $\xi$.
Consider a pair $(i,j) \in P$ and suppose $\xi(i,j) = e_t$ for some $t \in [m]$.
We want to verify that $\mathsf{parity}_i(\gamma_{e_t}) \neq \mathsf{parity}_j(\gamma_{e_t})$.
If $t = 1$, then $i$ and $j$ belong to different parts in $\varPhi_1 = \varPhi(\gamma_{e_1}) = \varPhi(\gamma_{e_t})$, i.e., $\mathsf{parity}_i(\gamma_{e_t}) \neq \mathsf{parity}_j(\gamma_{e_t})$.
If $t > 1$, then $i$ and $j$ belong to different parts in $\varPhi_t$ but belong to the same parts in $\varPhi_{t-1}$, which implies that $i$ and $j$ belong to different parts in $\varPhi(\gamma_{e_t})$, i.e., $\mathsf{parity}_i(\gamma_{e_t}) \neq \mathsf{parity}_j(\gamma_{e_t})$.
This completes the proof.
\end{proof}

With the above lemma in hand, we simply construct the family $\varXi$ in $k^{O(k)}$ time, and only try the functions in $\varXi$.
This improves the running time to $2^{O(p)} k^{O(k)} n^{O(k)}$, which is $2^{O(p)} n^{O(k)}$ because $k \leq n$ by our assumption.
\begin{theorem}
\textnormal{\sc Generalized Point-Separation} for connected obstacles in the plane can be solved in $2^{O(p)} n^{O(k)}$ time, where $n$ is the number of obstacles, $k$ is the number of points, and $p$ is the number of point-pairs to be separated.
\end{theorem}

\begin{corollary}
\textnormal{\sc Point-Separation} for connected obstacles in the plane can be solved in $2^{O(k^2)} n^{O(k)}$ time, where $n$ is the number of obstacles and $k$ is the number of points.
\end{corollary}

\section{An Improved Algorithm for Pseudo-disk Obstacles}
\label{sec:pseudodisks-alg}
In this section, we study \gptsep{} for \textit{pseudo-disk} obstacles and obtain an improved algorithm.
To this end, the key observation is the following analog of Lemma~\ref{lem-criterion} for pseudo-disk obstacles.
\begin{lemma} \label{lem-diskplanar}
Suppose $\cs$ consists of pseudo-disk obstacles.
Then a subset $\cs' \subseteq \cs$ is a $P$-separator iff there is a subgraph of the induced subgraph $G_\cs[\cs']$ that is planar and $P$-good.
\end{lemma}
\begin{proof}
The ``if'' part follows immediately from Lemma~\ref{lem-criterion}.
So it suffices to show the ``only if'' part.
Let $\cs' \subseteq \cs$ be a $P$-separator and $U = \bigcup_{S \in \mathcal{S}'} S$.
Recall that two obstacles $S,S' \in \mathcal{S}'$ \textit{contribute} to $U$ if an intersection point of the boundaries of $S$ and $S'$ is a break point on the boundary of $U$ (see Section~\ref{sec:prelims}).
By Fact~\ref{fact-planar}, the graph $G' = (\mathcal{S}',E)$ where $E = \{(S,S'): S,S' \in \cs' \text{ contribute to } U\}$ is planar.
We define a subgraph $G$ of the induced subgraph $G_\cs[\cs']$ as follows.
The vertex set of $G$ is $\cs'$.
For each edge $e = (S,S')$ of $G_\cs[\cs']$, if $S,S'$ contribute to $U$ or $S = S'$, then we include $e$ in $G$, otherwise we discard it.
We observe that $G'$ is planar.
Indeed, $G$ can be obtained from $G'$ by adding parallel edges and self-loops.
Since $G'$ is planar and adding parallel edges and self-loops does not change planarity, $G$ is also planar.
It now suffices to prove that $G$ is $P$-good.
Consider a pair $(i,j) \in P$ and we want to show the existence of a cycle $\gamma$ in $G$ such that $\mathsf{parity}_i(\gamma) \neq \mathsf{parity}_j(\gamma)$.
In the proof of Lemma~\ref{lem-criterion}, we constructed a cycle $\gamma$ in $G_\cs[\cs']$ satisfying $\mathsf{parity}_i(\gamma) \neq \mathsf{parity}_j(\gamma)$.
In that construction, the cycle $\gamma$ also satisfies the following property: for each pair $(S,S')$ of two consecutive vertices in $\gamma$, there are two \textit{adjacent} arcs of the boundary of $U$ contributed by $S$ and $S'$ respectively, which implies that $S,S'$ contribute to $U$.
Therefore, $\gamma$ is also a cycle in $G$.
It follows that $G$ is $P$-good, completing the proof.
\end{proof}

With the above lemma in hand, we are now ready to prove an analog of Lemma~\ref{lem-PPM2} for pseudo-disk obstacles.
The only difference is that here we can require $H^*$ to be planar.
\begin{lemma} \label{lem-diskplanarPPM}
Suppose $\cs$ is a set of pseudo-disk obstacles.
Then there exists a $P$-good $k$-labeled planar graph $H^*$ with at most $4k$ vertices and $5k$ edges and a PPM $f^*$ from $H^*$ to $G_\cs$ such that $\mathsf{cost}(f^*) = \mathsf{opt}$.
\end{lemma}
\begin{proof}
Recall that in the proof of Lemma~\ref{lem-PPM2}, we first took a minimal $P$-good subgraph $G$ of the induced subgraph $G_\cs[\cs']$, and then obtained $H^*$ by applying a path-contraction procedure on $G$.
The choice of $G$ is arbitrary as long as it is a minimal $P$-good subgraph of $G_\cs[\cs']$.
Furthermore, if $G$ is planar, then the resulting $H^*$ is also planar because the path-contraction procedure preserves planarity.
Therefore, it suffices to show that $G_\cs[\cs']$ has a minimal $P$-good subgraph that is planar.
By Lemma~\ref{lem-diskplanar}, there exists a $P$-good subgraph of $G_\cs[\cs']$ that is planar.
Since subgraphs of a planar graph are also planar, there exists a minimal $P$-good subgraph of $G_\cs[\cs']$ that is planar, which completes the proof.
\end{proof}

Now we explain how the planarity of $H^*$ in Lemma~\ref{lem-diskplanarPPM} helps us solve the problem more efficiently.
Recall how our algorithm in Section~\ref{sec-2^p} works.
We first enumerate the graph structure $H^* = (V_{H^*},E_{H^*})$ of $H^*$.
For a fixed graph structure, let $C_1,\dots,C_r$ be the connected components of $H^*$, and $T_1,\dots,T_r$ be spanning trees of $C_1,\dots,C_r$, respectively.
Let $E_0 \subseteq E_{H^*}$ be the set of edges that are not in $T_1,\dots,T_r$.
We then create the family $\varXi$ of functions $\xi: P \rightarrow E_0$ in Lemma~\ref{lem-representative}.
For each $\xi \in \varXi$, we use the algorithm of Lemma~\ref{lem-respect} to efficiently compute the ``optimal'' labeling-PPM pair $(\mathsf{lab},f^*)$ for $H^*$ satisfying the condition that $\mathsf{lab}$ respects $\xi$.
Here we apply the same framework, but replace Lemma~\ref{lem-respect} with an improved algorithm which works for the case that $H^*$ is planar.
The key ingredient of this improved algorithm is the planar separator theorem, which allows us to solve the problem of Lemma~\ref{lem-respect} more efficiently using divide-and-conquer when $H^*$ is planar.
\begin{lemma} \label{lem-planarrespect}
Suppose $H^*$ is planar.
Given $\xi:P \rightarrow E_0$, one can compute in $2^{O(p)} n^{O(\sqrt{k})}$ time a labeling-PPM pair $(\mathsf{lab},f^*)$ for $H^*$ which minimizes $\mathsf{cost}(f^*)$ subject to the condition that $\mathsf{lab}$ respects $\xi$.
\end{lemma}
\begin{proof}
As in the proof of Lemma~\ref{lem-respect}, suppose $E_{H^*} = \{e_1,\dots,e_m\}$ where $m = O(k)$.
Let $\delta: [m] \times E_0 \rightarrow \{0,1\}$ be an indicator defined as $\delta(t,e) = 1$ if $e_t$ is an edge of the cycle $\gamma_e$ and $\delta(t,e) = 0$ otherwise.
Consider a triple $(H,V',f_V')$, where $H = (V_H,E_H)$ is a subgraph of $H^*$, $V' \subseteq V_H$ is a subset of the vertex set of $H$, and $f_V': V' \rightarrow \mathcal{S}$ is a mapping.
For such a triple, we define a corresponding problem: for every function $\phi: P \rightarrow \{0,1\}$, computing a labeling-PPM pair $(\mathsf{lab},f)$ for $H$ (i.e., $\mathsf{lab}: E_H \rightarrow \{0,1\}^k$ is a labeling for the edges of $H$ and $f$ is a PPM from $H$ to $G_\mathcal{S}$ with respect to the labeling $\mathsf{lab}$) that minimizes $\mathsf{cost}(f)$ subject to \textbf{(i)} $f$ is compatible with $f_V'$, i.e., $f$ maps every $v \in V'$ to $f_V'(v)$ and \textbf{(ii)} $\sum_{e_t \in E_H} \delta(t,\xi(i,j)) \cdot (\mathsf{lab}_i(e_t) \oplus \mathsf{lab}_j(e_t)) = \phi(i,j)$.

We show how to solve the problem instance $(H,V',f_V')$ efficiently using divide-and-conquer.
Let $c$ be a sufficiently large constant.
If $|V_H| \leq c$, we simply solve the instance using brute-force in $O(1)$ time.
Assume $|V_H| > c$.
Since $H^*$ is planar, $H$ is also planar.
Thus, by the planar separator theorem, we can find in $|V_H|^{O(1)}$ time a partition of $V_H$ into three sets $V_1,V_2,X$ such that \textbf{(i)} there is no edge in $E_H$ between $V_1$ and $V_2$, \textbf{(ii)} $|X| \leq 3\sqrt{|V_H|}$, and \textbf{(iii)} $|V_1| \leq \frac{2}{3} |V_H|$ and $|V_2| \leq \frac{2}{3} |V_H|$.
We define two subgraphs $H_1$ and $H_2$ of $H$ as follows.
The graph $H_1 = (V_{H_1},E_{H_1})$ is the induced subgraph $H[V_1 \cup X]$, and the graph $H_2 = (V_{H_2},E_{H_2})$ is defined as $V_{H_2} = V_2 \cup X$ and $E_{H_2} = E_H \backslash E_{H_1}$.
Observe that $H_1$ and $H_2$ cover all the vertices and edges of $H$.
In addition, $H_1$ and $H_2$ share the common vertices in $X$ and do not share any common edges.
Let $V_1' = (X \cup V') \cap V_{H_1}$ and $V_2' = (X \cup V') \cap V_{H_2}$.
We enumerate all functions $g : X \rightarrow \cs$ that compatible with $f_V'$, i.e., $g(v) = f_V'(v)$ for all $v \in X \cap V'$.
The number of such functions is $n^{O(\sqrt{|V_H|})}$ because $|X| = O(\sqrt{|V_H|})$.
For a fixed function $g : X \rightarrow \mathcal{S}$, let $g': X \cup V' \rightarrow \cs$ be the function obtained by gluing $g$ and $f_V'$, i.e., $g'(v) = g(v)$ on $X$ and $g'(v) = f_V'(v)$ on $V'$.
We then recursively solve the two problem instances $\mathsf{Prob}_{g,1} = (H_1,V_1',g_1')$ and $\mathsf{Prob}_{g,2} = (H_2,V_2',g_2')$ where $g_1'$ (resp., $g_2'$) is the function obtained by restricting $g'$ to $V_1'$ (resp., $V_2'$).
After all functions $g : X \rightarrow \cs$ are considered, we collect all the solutions for the problem instances $\mathsf{Prob}_{g,1}$ and $\mathsf{Prob}_{g,2}$.

We are going to use these solutions to obtain the solution for the problem instance $(H,V',f_V')$.
Recall that for every function $\phi: P \rightarrow \{0,1\}$, we want to compute a labeling-PPM pair $(\mathsf{lab},f)$ for $H$ that minimizes $\mathsf{cost}(f)$ subject to \textbf{(i)} $f$ is compatible with $f_V'$ and \textbf{(ii)} $\sum_{e_t \in E_H} \delta(t,\xi(i,j)) \cdot (\mathsf{lab}_i(e_t) \oplus \mathsf{lab}_j(e_t)) = \phi(i,j)$.
We first guess how the desired PPM $f$ maps the vertices in $X$, which can be described as a function $g: X \rightarrow \mathcal{S}$.
There are in total $n^{O(\sqrt{|V_H|})}$ guesses we need to make.
Now suppose our guess for $g$ is correct.
As before, we define $g': X \cup V' \rightarrow \cs$ as the function obtained by gluing $g$ and $f_V'$.
Note that $f$ is compatible with $g'$.
Let $(\mathsf{lab}_1,f_1)$ and $(\mathsf{lab}_2,f_2)$ denote labeling-PPM pairs for $H_1$ and $H_2$, respectively, obtained by restricting $(\mathsf{lab},f)$ to $H_1$ and $H_2$.
Define $\phi_1: P \rightarrow \{0,1\}$ as $\phi_1(i,j) = \sum_{e_t \in E_{H_1}} \delta(t,\xi(i,j)) \cdot (\mathsf{lab}_i(e_t) \oplus \mathsf{lab}_j(e_t))$ and $\phi_2: P \rightarrow \{0,1\}$ as $\phi_2(i,j) = \sum_{e_t \in E_{H_2}} \delta(t,\xi(i,j)) \cdot (\mathsf{lab}_i(e_t) \oplus \mathsf{lab}_j(e_t))$.
We observe that $f_1$ and $f_2$ must be compatible with $g_1'$ and $g_2'$, respectively, where $g_1'$ (resp., $g_2'$) is the function obtained by restricting $g'$ to $V_1'$ (resp., $V_2'$), because $f$ is compatible with $g'$.
Also, we have $\phi = \phi_1 \oplus \phi_2$ and $\mathsf{cost}(f) = \mathsf{cost}(f_1) + \mathsf{cost}(f_2) - |X|$, because $V_{H_1} \cap V_{H_2} = X$ and $\{E_{H_1},E_{H_2}\}$ is a partition of $E_H$.
On the other hand, as long as $f_1$ and $f_2$ are compatible with $g_1'$ and $g_2'$ respectively and $\phi = \phi_1 \oplus \phi_2$, we can always glue the two labeling-PPM pairs $(\mathsf{lab}_1,f_1)$ and $(\mathsf{lab}_2,f_2)$ to obtain a labeling-PPM pair $(\mathsf{lab},f)$ for $H$ satisfying $\mathsf{cost}(f) = \mathsf{cost}(f_1) + \mathsf{cost}(f_2) - |X|$ such that \textbf{(i)} $f$ is compatible with $g'$ and \textbf{(ii)} $\sum_{e_t \in E_H} \delta(t,\xi(i,j)) \cdot (\mathsf{lab}_i(e_t) \oplus \mathsf{lab}_j(e_t)) = \phi(i,j)$.
Therefore, we can solve the problem as follows.
We simply guess the functions $\phi_1$ and $\phi_2$ satisfying $\phi_1 \oplus \phi_2 = \phi$.
There are in total $2^p$ guesses we need to make.
Suppose our guess is correct.
We retrieve the solution $(\mathsf{lab}_1,f_1)$ of the problem instance $\mathsf{Prob}_{g,1}$ for the function $\phi_1$ and the solution $(\mathsf{lab}_2,f_2)$ of the problem instance $\mathsf{Prob}_{g,2}$ for the function $\phi_2$, which have already been computed.
We know that $(\mathsf{lab}_1,f_1)$ (resp., $(\mathsf{lab}_2,f_2)$) minimizes $\mathsf{cost}(f_1)$ (resp., $\mathsf{cost}(f_2)$) subject to \textbf{(i)} $f_1$ is compatible with $g_1'$ (resp., $f_2$ is compatible with $g_2'$) and \textbf{(ii)} $\sum_{e_t \in E_{H_1}} \delta(t,\xi(i,j)) \cdot (\mathsf{lab}_i(e_t) \oplus \mathsf{lab}_j(e_t)) = \phi_1(i,j)$ (resp., $\sum_{e_t \in E_{H_2}} \delta(t,\xi(i,j)) \cdot (\mathsf{lab}_i(e_t) \oplus \mathsf{lab}_j(e_t)) = \phi_2(i,j)$).
By gluing $(\mathsf{lab}_1,f_1)$ and $(\mathsf{lab}_2,f_2)$, we obtain a labeling-PPM pair $(\mathsf{lab},f)$ for $H$, which is what we want because of the optimality of $(\mathsf{lab}_1,f_1)$ and $(\mathsf{lab}_2,f_2)$ and the fact that our guesses for $g$ and $\phi_1,\phi_2$ are all correct.

Finally, we analyze the running time of the above algorithm.
Let $T(h)$ denote the time cost for solving a problem instance $(H,V',f_V')$ with $|V_H| = h$.
We have $T(h) = O(1)$ for $h \leq c$, because we use brute-force for the case $h \leq c$.
Suppose $h > c$.
In this case, we have recursive calls on the subgraphs $H_1$ and $H_2$ of $H$.
Note that $H_1 = |V_1| + |X| \leq \frac{2}{3}h + 3\sqrt{h} \leq \frac{3}{4}h$, because $h>c$ and $c$ is sufficiently large.
Similarly, we have $H_2 \leq \frac{3}{4}h$.
The number of recursive calls is $n^{O(\sqrt{h})}$.
Besides the recursive calls, all work can be done in $2^{O(p)} n^{O(\sqrt{k})}$ time.
Therefore, we have the recurrence $T(h) = n^{O(\sqrt{h})} \cdot T(\frac{3}{4} h) + 2^{O(p)} n^{O(\sqrt{h})}$, which solves to $T(h) = 2^{O(p)} n^{O(\sqrt{h})}$.
To solve the problem of the lemma, the initial call is for the problem instance $(H^*,\empty,\mathsf{null})$, which takes $2^{O(p)} n^{O(\sqrt{k})}$ time since $|V_{H^*}| = O(k)$.
\end{proof}

Replacing Lemma~\ref{lem-respect} with Lemma~\ref{lem-planarrespect}, we can apply the algorithm in Section~\ref{sec-2^p} to solve the generalized point-separation problem in $2^{O(p)} k^{O(k)} n^{O(\sqrt{k})}$ time.
\begin{theorem}
\textnormal{\sc Generalized Point-Separation} for pseudo-disk obstacles in the plane can be solved in $2^{O(p)} k^{O(k)} n^{O(\sqrt{k})}$ time, where $n$ is the number of obstacles, $k$ is the number of points, and $p$ is the number of point-pairs to be separated.
\end{theorem}

\begin{corollary}
\textnormal{\sc Point-Separation} for pseudo-disk obstacles in the plane can be solved in $2^{O(k^2)} n^{O(\sqrt{k})}$ time, where $n$ is the number of obstacles and $k$ is the number of points.
\end{corollary}

\label{sec:hardness}

\section{\textsf{ETH}-Hardness of Points-Separation}
\label{sec:eth-hardness}
In the previous sections, we gave an $f(k) \cdot n^{O(k)}$-time algorithm for $k$-\pointsep with general (connected) obstacles
and an $f(k) \cdot n^{O(\sqrt{k})}$-time algorithm with pseudo-disk obstacles.
In this section, we show that assuming Exponential Time Hypothesis (ETH), both of our algorithms are almost tight and significant improvement is unlikely. We begin by describing our reduction for general obstacles. 

\subsection{Hardness for General Obstacles}
	We give a reduction from \textsc{Partitioned Subgraph Isomorphism} (PSI) problem which is defined as follows.
Recall that in the \textsc{Subgraph Isomorphism} problem, we are given two
graphs $G$ and $H$ and we want to find an injective mapping $\psi: V(G) \rightarrow V(H)$ such that if $(u, v) \in E(G)$, 
then $(\psi(u), \psi(v)) \in E(H)$. In the \textsc{Partitioned Subgraph Isomorphism} problem, we want to find a \emph{colorful mapping} 
of $G$ into $H$. Formally, we are given undirected graphs $H$ and $G$ where $G$ has maximum degree $3$, and a \emph{coloring} 
function $\col : V(H) \rightarrow V(G)$ that partitions vertices of $H$ into $|V(G)|$ classes. 
We say that an injective mapping $\psi: V(G) \rightarrow V(H)$ is a \emph{colorful mapping} of $G$ into $H$, 
if for every $v \in V(G)$, $\col(\psi(v)) = v$, and for every $(u, v) \in E(G)$, we have $(\psi(u), \psi(v)) \in E(H)$.
Then in the \textsc{Partitioned Subgraph Isomporphism}, we want to find if there exists a colorful mapping
of $G$ into $H$. 

We will use the following well-known result of Marx~\cite{marx2007can} relevant to our reduction.
\begin{theorem}{\textnormal{\cite[Corollary~6.3]{marx2007can}}}
\label{thm:marx}
Unless \textsf{ETH} fails, \textsc{PSI} cannot be solved in $f(k) n^{o(k/\log k)}$ time for
any function $f$ where $k = |E(G)|$ and $n= |V(H)|$.
\end{theorem}

\paragraph*{Our Construction.}
Given an instance of PSI as graphs $G, H$ and coloring $\col: V(H) \rightarrow V(G)$,
we want to construct an instance of \pointsep, namely a set of 
obstacles $\cs$ and a set of points $A$ such that all point pairs in $A$ are separated. 
For the ease of exposition, we will first discuss a reduction from PSI to an instance $(\cs, A, P)$ of
\gptsep where the set $P$ of \emph{request pairs} is specified. Later we extend the construction to
show that the same bounds also hold for \pointsep.

The set of obstacles $\cs$ used in our construction mainly consists of
an obstacle $S_{pq}$ for every edge $(u_p, u_q) \in E(H)$. In addition, we also use an additional auxiliary obstacle
denoted by $S_0$. All the obstacles and request pairs will be contained in a rectangle $\bbox$ with 
bottom-left corner $(0, 0)$ and top-right corner $(z, 3)$, where $z$ is the total number of request pair \emph{groups}.
Each group can have at most two request pairs.
We split the rectangle $\bbox$ into $z$ blocks, each of width one. 
The $r$-th block $B_r$ is bounded by the vertical lines $x={r-1}$ and $x=r$, contains the $r$-th
request pair group. Initially all obstacles are horizontal line segments of length $z$ occupying the part of $x$-axis 
from $x=0$ to $x=z$ and coincident to the bottom side of $\bbox$.
Moreover, let $\ell_1, \ell_2$ be two horizontal line segments coincident with $y=1$ and $y=2$ respectively
and starting from $x=0$ (left boundary of $\bbox$) and ending at $x=z$ (right boundary of $\bbox$).
These line segments will serve as \emph{guardrails} for obstacle growth. Specifically obstacles can
only grow vertically at $x=r$ (for some integer $r$) or horizontally along the lines $\ell_1, \ell_2$.
(See also Figure~\ref{fig:hardness-example}.)

\begin{figure}[htb!]
	\centering
	\includegraphics{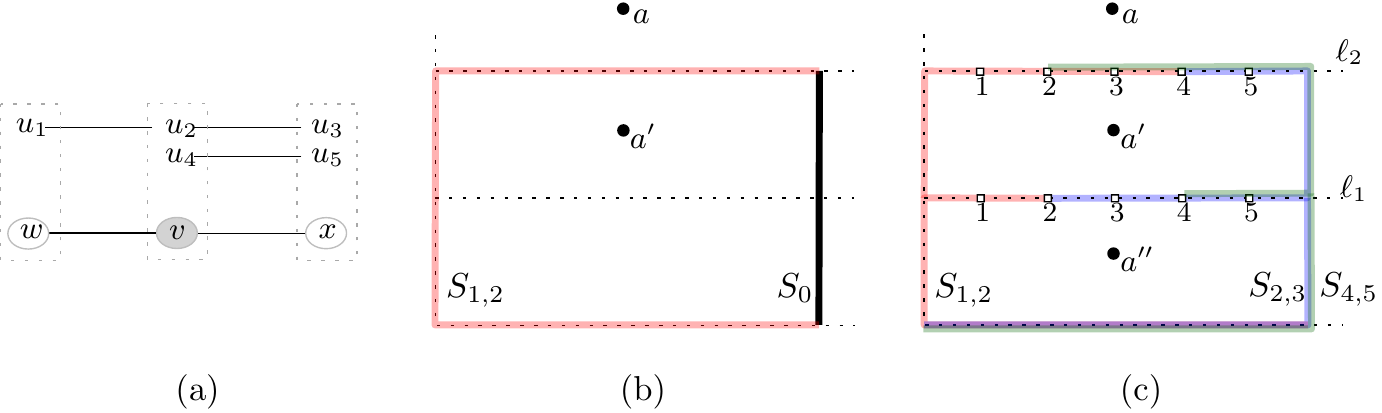}
	\caption{An example construction. (a) Graphs $G$ and $H$ with the $\col : V(H) \rightarrow V(G)$ shown by
					dotted boxes around nodes. (b)~ Block with \tone request pair for edge $(v, w) \in E(G)$.~  (c)~ Block
					with \ttwo request pair group for vertex $v$ and its adjacent edges $(v, w), (v, x) \in E(G)$. 
					Obstacles $S_{1,2}$ and $S_{2,3}$ separate both pairs $(a, a')$ and $(a', a'')$ 
					whereas $S_{1,2}$ and $S_{4,5}$ does not.}
	\label{fig:hardness-example}
\end{figure}

The $r$-th request pair group is contained in block $B_r$
and may consist of points $a_r, a'_r, a_r''$ where
$a_r = (r-\frac{1}{2}, \frac{5}{2})$, 
$a_r' = (r-\frac{1}{2}, \frac{3}{2})$
and $a_r'' = (r-\frac{1}{2}, \frac{1}{2})$.
We have two types of groups: \emph{\tone request pair group} 
consisting of \emph{one} request pair $(a_r, a_r')$ and
\emph{\ttwo request pair group} 
consisting of \emph{two} request pairs $p_r = (a_r, a_r')$ and $p_r' = (a_r', a_r'')$.
Depending on the type of the group, we will now \emph{grow} the obstacles in a systematic 
manner so that they interact in the neighborhood of request pairs. 
\begin{enumerate}
	\item \emph{\tone request pair group}~ For every edge $e_i = (v, w) \in E(G)$, we add a request pair $p_r = (a_r, a_r')$ 
			to $P$. Next we grow the obstacles around $p_r$ as follows.
			(See also Figure~\ref{fig:hardness-example}b.)
			\begin{itemize}
				\item Extend the auxiliary obstacle $S_0$ vertically along $x=r$ until $y=2$.
				\item For every $(u_p, u_q) \in E(H)$ such that $(\col(u_p), \col(u_q)) = e_i$, extend
							the obstacle $S_{pq}$ vertically along $x={r-1}$ until $y=2$ and then rightwards
							along $\ell_2$ until it touches $S_0$.
			\end{itemize}
				Observe that to separate \tone request pair $p_r$, we must select $S_0$ and one obstacle corresponding to
				an edge of $H$.

	\item \emph{\ttwo request pair group} For a vertex $v \in V(G)$ and pair of edges 
			$e_i, e_j \in E(G)$ adjacent to $v$ with $i < j$, we add two request pairs $p_r = (a_r, a_r')$
			and $p_r' = (a_r', a_r'')$ to $P$. In order to grow the obstacles, consider the unit length intervals along 
			lines $\ell_1, \ell_2$ contained in $B_r$. We subdivide these intervals by adding $n$ markers each 
			separated by a small distance $\eps = \frac{1}{n+1}$. Here $n=|V(H)|$. 
		  We will use these markers to define the precise boundary of obstacles in block $B_r$. 
			(See also Figure~\ref{fig:hardness-example}c.)
	\begin{itemize}
		\item Let $e_i = (v, w)$ and $S_{pq} = (u_p, u_q)$ be an obstacle such that $(\col(u_p), \col(u_q)) = e_i$.
					Without loss of generality, assume that $\col(u_p) = v$ and $\col(u_p) = w$.
					First we extend $S_{pq}$ along the \emph{left} boundary of $B_r$ along $x=r-1$ until $y=2$.
					Then we connect $S_{pq}$ to marker $p$ along line $\ell_1$ and to marker $n - p + 1$ along line $\ell_2$, moving 
					from \emph{left to right}.
		\item Similarly, let $e_j = (v, x)$ and $S_{gh} = (u_g, u_h)$ be an obstacle such that $(\col(u_g), \col(u_h)) = e_j$.
					Without loss of generality, assume that $\col(u_g) = v$ and $\col(u_h) = x$.
					We extend $S_{gh}$ along the \emph{right} boundary of $B_r$ along $x=r$ until $y=2$.
					Then we connect $S_{gh}$ to marker $g$ along line $\ell_1$ and to marker $n - g + 1$ along line $\ell_2$,
					moving from \emph{right to left}.
	\end{itemize}
	Observe that to separate both \ttwo request pairs $p_r$ and $p_r'$, we must select two obstacles corresponding to edges of $H$.
\end{enumerate}

It is easy to verify that all the obstacles are simple and connected. Observe that since each vertex has
maximum degree $3$, the total number of request pairs added is $z \leq |E(G)| + 2 \cdot 3 |V(G)| = O(k)$ where $k = V(G)$.
The total number of obstacles $|\cs| = |E(H)| + 1 = O(n^2)$ where $n = |V(H)|$.
\begin{observation}
	\label{obs:numObstacles} For the \gptsep instance $(\cs, A, P)$ constructed above, 
	we have $|\cs| = O(n^2)$, $|A| = O(k)$ and $|P| = O(k)$.
\end{observation}

We prove the following lemma which will be useful later.
\begin{lemma}
	\label{lemma:blockSep} 
	Let $p_r = (a_r, a_r')$ and  $p_r' = (a_r', a_r'')$ be a \ttwo request pair group corresponding to vertex $v$ and its two 
	adjacent edges $e_i = (v, w)$ and $e_j = (v, x)$ such that $i < j$. 
	Then two obstacles $S_{pq}$ defined by the edge $(u_p, u_q)$ and  $S_{gh}$ defined by $(u_g, u_h)$
	separate both $p_r$ and $p_r'$ if and only if $p = g$ and 
	$\col(u_p) = \col(u_g) = v$, $\col(u_q) = w$, $\col(u_h) = x$.
\end{lemma}
\begin{proof}
	The reverse direction is easy to verify. Specifically, if 
	$\col(u_p) = \col(u_g) = v$, $\col(u_q) = w$, $\col(u_h) = x$
	then the obstacles $S_{pq}$ and $S_{gh}$ are respectively coincident with left and right 
	boundary of block $B_r$. Moreover, since $p = g$, both the obstacles overlap precisely at 
	marker $p$ along $\ell_1$ and $n-p+1$ along $\ell_2$, forming a closed 
	curve containing only point $a_r' = (r-1, \frac{3}{2})$. Therefore, both the pairs $p_r$
	and $p_r'$ are separated.

	For the other direction, from the way obstacles $S_{pq}$ and $S_{gh}$ interact in block $B_r$:
	they may overlap along $\ell_1$ or $\ell_2$ or both or none.
	If the obstacles overlap only along $\ell_1$, they cannot
	separate pair $p_r$. Similarly, if they overlap only along $\ell_2$, they 
	cannot separate the pair $p_r'$. Since both pairs are separated, obstacles 
	$S_{pq}$ and $S_{gh}$ must overlap along both $\ell_1$, $\ell_2$ and 
	form a closed curve containing point $a_r'$. 
	This can only happen if $S_{pq}, S_{gh}$ overlap in block $B_r$ approaching  $\ell_1, \ell_2$
	from opposite sides. Without loss of generality, we can assume that $S_{pq}$ is coincident with left boundary 
	of $B_r$ and $S_{gh}$ is coincident with the right boundary of $B_r$. 
	This can happen only if $\col(u_p) = \col(u_g) = v$, $\col(u_q) = w$, $\col(u_h) = x$.
	It remains to show that $p = g$. Observe that since $S_{pq}, S_{gh}$ overlap on $\ell_1$, we must have that
	marker $p$ is to the right of marker $g$. That is $p \geq g$. Similarly, since $S_{pq}, S_{gh}$ overlap on $\ell_2$,
	we have $n - p + 1 \geq n - g + 1$ which gives $p \leq g$. Combining these, we get $p = g$.
\end{proof}

	We now prove the following lemma that establishes the correctness of our reduction.
\begin{lemma}
  \label{lemma:correctness}
	Given an instance of PSI as graphs $G, H$ and coloring, $\col: V(H) \rightarrow V(G)$, there exists
	a colorful mapping $\psi: V(G) \rightarrow V(H)$ if and only
	if the point pairs $P$ can be separated by a set of $m = |E(G)| + 1$ obstacles $\cs^* \subseteq \cs$.
\end{lemma}
\begin{proof}
	$(\Rightarrow)$~ Given a colorful mapping $\psi$ we construct the set of obstacles $\cs^*$ as follows.
	For every edge  $e = (v, w) \in E(G)$, include the obstacle $(\psi(v), \psi(w))$ to $\cs^*$ -- such an 
	obstacle always exists because $(\psi(v), \psi(w)) \in E(H)$. Next, we add $S_0$ to $\cs^*$.
	It is easy to verify that $\cs^*$ separates the \tone request pairs. 
	For a \ttwo request pair group $p_r, p_r'$ at vertex $v$ and edges $e_i = (v, w)$, $e_j = (v, x)$, 
	let $u_p = \psi(v), u_q = \psi(w)$ and $u_h = \psi(x)$. Since $\psi$ is a colorful mapping,
	we have $\col(u_p) = \col(\psi(v)) = v$. Similarly, $\col(u_q) = w$ and $\col(u_h) = x$.
	Therefore, it follows from Lemma~\ref{lemma:blockSep} that $\cs^*$ separates request
	pairs $p_r, p_r'$, for all $1 \leq r \leq z$.

	$(\Leftarrow)$~ Given a set $\cs^*$ of $m$ obstacles that separates all request pairs,
	we will first construct an injective function $\cm: E(G) \rightarrow E(H)$ that uniquely
	maps every edge of $G$ to an edge of $H$. Consider the set $P_1$ of \tone request pairs.
	Since $\cs^*$ separates $P_1$, it must include $S_0$ and a unique obstacle $S_{pq} = (u_p, u_q)$
	for every edge $e_i = (v, w) \in E(G)$ such that $(\col(u_p), \col(u_q)) = e_i$. 
	The uniqueness of $S_{pq}$ follows from the fact that there are $|E(G)|$ \tone request pairs
	and $|\cs^*| = |E(G)| + 1$ obstacles. We assign $\cm(e_i) = (u_p, u_q)$. 

	Next, we build a colorful mapping $\psi$ that is consistent with the mapping $\cm$ of edges. 
	For this, we use the fact that $\cs^*$ also separates \ttwo request pair groups. 
	Consider the \ttwo request pair group corresponding to vertex $v \in V(G)$ and edges 
	$e_i = (v, w)$ and $e_j = (v, x)$ with $i < j$. We apply Lemma~\ref{lemma:blockSep} over this
	group with obstacles defined by edges $(u_p, u_q) = \cm(e_i)$, and $(u_g, u_h) = \cm(e_j)$.
	This gives $u_p = u_g$ and $\col(u_p) = v$.
	Since this holds for every pair of edges $e_i, e_j$ adjacent to vertex $v$, 
	we can assign $\psi(v) = u_p$, which also satisfies $\col(\psi(v)) = \col(u_p) = v$
	required for a colorful mapping. 
	Repeating this for every $v$ gives the complete mapping $\psi : V(G) \rightarrow V(H)$.
	It remains to show that if $(v, w) \in E(G)$, then $(\psi(v), \psi(w)) \in E(H)$.
	To see this, observe that for every $e_i = (v, w) \in E(G)$ the edge $(u_p, u_q) = \cm(e_i)$
	exists in $E(H)$, or else we would not be able to separate the \tone request pair for $e_i$.
	From the way we assign $\psi(v)$, it follows that $\psi(v) = u_p$ and $\psi(w) = u_q$.
	Therefore, $(\psi(u_p), \psi(u_q)) \in E(H)$.
\end{proof}

	We will now extend the above construction $(\cs, A, P)$ to
	the special case when $P$ consists of all pairs of points in $A$. We do this by adding $z$ special 
	obstacles called \emph{barriers} (one for each block $B_r$) and one \emph{master} point $a_0 = (0, 4)$ that lies to the
	outside of rectangle $\bbox$ enclosing all obstacles.
	Each barrier $S_r$ around block $B_r$ is an inverted U-shaped obstacle that is coincident with the
	left, top and bottom boundaries of $B_r$. More precisely, obstacle $S_r$ consists of three segments:
	a vertical segment from $(r-1, 0)$ to $(r-1, 3)$, a horizontal segment from $(r-1, 3)$ to $(r, 3)$
	and then a vertical segment from $(r, 3)$ to $(r, 0)$. 
	
	Let $\cs_b$ be the set of all barrier obstacles added above, we prove the following lemma.
	\begin{lemma}
		\label{lemma:extension}
		There exists a solution with $|E(G)|+1$ obstacles for the \gptsep instance $(\cs, A, P)$ constructed before 
		if and only if there exists a solution with $|E(G)|+1+|\cs_b|$  obstacles for the
		\pointsep instance $(\cs \cup \cs_b,~ A \cup a_0)$.
	\end{lemma}
	\begin{proof}
		$(\Rightarrow)$ Add all the barriers $\cs_b$ to the solution for \gptsep. 
		All points that lie in the
		same block are already separated. Any pair of points that lie in different blocks are separated due to
		the barrier obstacles $\cs_b$, which also separate every point in $A$ from the master point $a_0$.

		$(\Leftarrow)$
			The only way to separate the master point $a_0$ from point $a_r$ in block $B_r$ is to select the 
			corresponding barrier $S_r$. Therefore, every solution must select all obstacles in $\cs_b$. 
			Since the set $\cs_b$ does not separate any within-block request pair, the remaining 
			set of $|E(G)|+1$ non-barrier obstacles must separate all request pairs in $P$.
	\end{proof}

Using Lemma~\ref{lemma:extension} along with Lemma~\ref{lemma:correctness}, 
Observation~\ref{obs:numObstacles} and applying Theorem~\ref{thm:marx}, we
obtain the following result for \pointsep.
\begin{theorem}
	\label{thm:eth-hardness}
	Unless \textsf{ETH} fails, a \pointsep instance $(\cs, A)$ cannot be solved in $f(k) n^{o(k/\log k)}$ time
	where $n = |\cs|$ and $k = |A|$.
\end{theorem}

\subsection{Hardness for Pseudodisk Obstacles}
For the case of pseudodisk obstacles, we will give a reduction from \textsc{Planar Multiway Cut} problem:
given an undirected \emph{planar} graph $G$ with a subset of $k$ vertices specified as terminals, the task is to
find a set of edges having minimum total weight whose deletion pairwise separates the $k$ terminal
vertices from each other. We will use another result by Marx~\cite{marx2012tight} which showed that unless
ETH fails, \textsc{Planar Multiway Cut} cannot be solved in $f(k) \cdot n^{o(\sqrt{k})}$ time. The 
result also holds when each edge has unit weight, which is the case we will reduce from.

\paragraph*{Our Construction} We first fix an embedding of the planar graph $G$ and consider its dual graph $G^*$.
Then we create an instance $(\cs, A)$ of \pointsep as follows. (See also Figure~\ref{fig:hardness-pseudodisks-example}.)
\begin{figure}[htb!]
	\centering
	\includegraphics{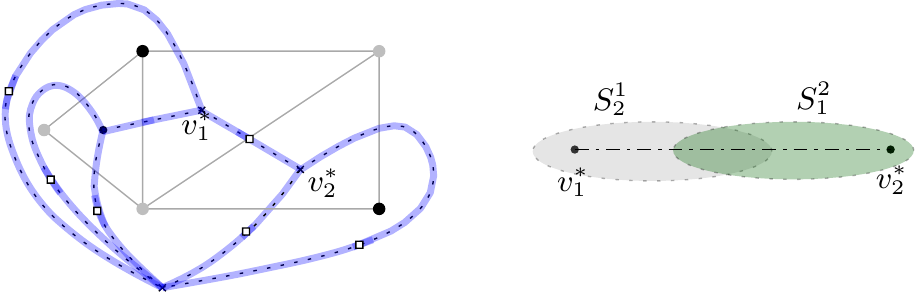}
	\caption{An example construction with pseudodisks. (a) The primal graph $G$ and dual graph $G^*$ are shown. The obstacles $\cs$
				move along the dual edges and overlap at the square \emph{markers}. The terminals of $G^*$ which form the point set
				$A$ are shown in bold. ~(b) An illustration of how the two obstacles for the dual edge $(v_1^*, v_2^*)$ 
				overlap is shown enlarged for clarity.}
	\label{fig:hardness-pseudodisks-example}
\end{figure}

\begin{itemize}
	\item \emph{Adding obstacles.~} For every edge $e_{ij}^* = (v_i^*, v_j^*) \in E(G^*)$, 
		we add two obstacles $S^i_j, S^j_i$ such 
		that $S^i_j$ encloses the dual vertex $v^*_i$ and extends halfway along $e_{ij}^*$.
		Similarly, $S^j_i$ encloses the dual vertex $v^*_j$ and extends halfway along $e_{ij}^*$ until it meets obstacle $S^i_j$. 
	\item \emph{Adding points.~} For each terminal $t_i$, which is a vertex of the primal graph $G$, add a point $a_i$ with
				same coordinates as that of $t_i$ in the embedding.
\end{itemize}

Observe that any pair of obstacles either overlap at their source vertex or at the middle of an edge, but not at both places. 
Therefore, no pair of obstacles intersect more than once and the construction can be realized with only pseudodisk obstacles.
The following lemma establishes the correctness of our reduction.
\begin{lemma}
	There exists a solution to \textsc{Planar Multiway Cut} with $m$ edges if and only if the \pointsep instance constructed above
	has a solution of size $2m$.
\end{lemma}
\begin{proof}
	For the forward direction, consider any pair of terminals $t_x, t_y$ -- since they are separated by the cut edges $E_c$, there must be
	a cycle in the dual graph separating $t_x, t_y$ and only consisting of dual of cut edges $E_c^*$. Repeating this for every pair of
	terminals gives a family of separating cycles consisting only of edges $E_c^*$. 
	It is easy to verify that replacing each dual edge $e^*_{ij}$ with its obstacle pair $S^i_j, S^j_i$ will also separate every 
	point pair corresponding to the terminals.

	For the other direction, given a solution $\cs'$ for \pointsep, we can draw curves in the plane that separate every point pair and lie in 
	the union of $\cs'$. We can assume that the solution is exclusion-wise minimal, so every time we arrive inside an obstacle at 
	vertex $v_i^*$, we must continue along an edge $e^*_{ij}$ where we must transfer to the other sibling obstacle $S^j_i$ for $e^*_{ij}$.
	Using these dual edges, we can construct a solution to \textsc{Planar Multiway Cut} of cost $|\cs'|/2$.
\end{proof}

Since \textsc{Planar Multiway Cut} cannot be solved in $f(k) n^{o(\sqrt{k})}$ time assuming ETH , we obtain the following result.
\begin{theorem}
	\label{thm:eth-hardness-pseudodisks}
	Unless \textsf{ETH} fails, a \pointsep instance $(\cs, A)$ with pseudodisk obstacles 
	cannot be solved in $f(k) n^{o(\sqrt{k})}$ time where $n = |\cs|$ and $k = |A|$.
\end{theorem}
It is not difficult to see that the above construction can also be realized using only unit disks.
In particular, we can replace each pseudodisks with a \emph{chain of unit disks} and achieve the same result.

\section{Hardness of Approximation}
\label{sec:apx-hardness}
	We will now switch our focus from exact algorithms to approximation algorithms for
	\pointsep with obstacles $\cs$ and input points $A$.
	Gibson et al.~\cite{constApxPseudodisks} gave a constant factor approximation algorithm for \pointsep 
	when obstacles are pseudodisks. However, not much is known for more general obstacle shapes, other than
	a factor $O(|A|)$-approximation that readily follows from the natural extension of their algorithm for pseudodisks. 
	In this section, we show that assuming the so-called \dvrconj conjecture, 
	\pointsep is significantly harder to approximate for general obstacle shapes.
	In particular, we show that assuming \dvrconj, it is not possible to approximate
	\pointsep within a factor $|A|^{1/2-\eps}$ or $|\cs|^{3-2\sqrt{2}-\eps}$ for any $\eps > 0$.
	
	We begin by first stating \dvrconj{}, a well-known complexity-theoretic assumption 
	about the hardness for the densest $k$-subgraph problems.
\begin{Conjecture}[\dvrconj~\cite{chlamtavc2017minimizing}]\label{conj:dvr}
	For all $0 < \alpha,\beta < 1$ with $\beta < \alpha - \eps$ for sufficiently small $\eps > 0$, 
	and function $k : \mathbb{N} \rightarrow  \mathbb{N}$ so that $k(n)$ grows polynomially 
	with $n$, $(k(n))^{1+\beta} \leq n^{(1+\alpha)/2}$, there does not exist an algorithm {\sf ALG} that takes as input an $n$-vertex graph $G$, 
	runs in polynomial time, and outputs either {\sf dense} or {\sf sparse}, such that:
	\begin{itemize}\setlength\itemsep{-.7mm}
	\item For every graph $G$ that contains an induced subgraph on $k = k(n)$ vertices and $k^{1+\beta}$ edges, 
				${\sf ALG}(G)$ outputs {\sf dense} with high probability.
	\item If $G$ is drawn from $G(n, p)$ with $p = n^{\alpha - 1}$ then ${\sf ALG}(G)$ outputs {\sf sparse} with high probability.
	\end{itemize}
\end{Conjecture}

	The conjecture was originally stated in~\cite{chlamtavc2017minimizing} but the formalization of the conjecture
	as stated above is borrowed from~\cite{mcp2021}. In order to obtain hardness guarantees for our problem using 
	Conjecture~\ref{conj:dvr}, we will describe a reduction that given a graph $G$ constructs
	an instance of \pointsep. Then we show that the images of dense instances under this reduction 
	will have (with high probability) optimum {\em at most} $x^*_d$, whereas the images of random 
	instances will have optimum at least $x^*_r$, where $x^*_r$ is much bigger than $x^*_d$. 
	Let $\rho = x^*_r/ x^*_d $ be the 
	\emph{distinguishing ratio} of the reduction, then an approximation algorithm for \pointsep with 
	ratio smaller than $\rho$ can now (with high probability) distinguish between the images of dense 
	and random instances, thereby refuting Conjecture~\ref{conj:dvr}. This gives us the following lemma.

	\begin{lemma}\label{lemma:distinguishing-ratio}
	If there exists a reduction with distinguishing ratio $\rho$, then, assuming \dvrconj, 
	there is no polynomial time approximation algorithm for \pointsep with approximation ratio less than $\rho$.
	\end{lemma}
	
	Our construction is inspired from a similar construction using \dvrconj for the related \textsc{Min-color Path}
	problem from~\cite{mcp2021}. Specifically, we borrow the idea of \emph{partitioning} the edges of graph $G = (V, E)$
	into $z$ groups $E_1, E_2, \dots, E_z$, by assigning every edge to one of the groups with probability $1/z$ independent
	of other edges. We have the following lemma.
	\begin{lemma}[Lemma~7.3~\cite{mcp2021}]
		\label{lemma:groupsPartition}
		For any graph $G=(V, E)$, there exists a partitioning of edges into $z = \frac{q} {2 \ln n}$ groups such that
		for any set $E^* \subseteq E$ of $q$ edges, every group $E_i \in \{E_1, E_2, \dots, E_z\}$
		contains an edge from $E^*$.
	\end{lemma}

	We will also need the following bound on the size of a subgraph of $G(n, p)$.
	\begin{lemma}[Lemma~7.2~\cite{mcp2021}]
		\label{lemma:random-subgraph}
		Let $G$ be drawn from $G(n, p)$. Then, with high probability, every 
		subgraph of $G$ with $q = n^{\Omega(1)}$ edges contains $\tilde{\Omega}(\min\{q, \sqrt{(q/p)}\})$ vertices. 
		Here $\tilde{\Omega}$ ignores logarithmic factors.
	\end{lemma}

	\paragraph*{Our Construction} Given a graph $G=(V, E)$ and fixed $\alpha, \beta$ and function $ k : \mathbb{N} \rightarrow  \mathbb{N}$
	satisfying conditions of Conjecture~\ref{conj:dvr}, we will construct an instance of \pointsep as follows.
	\begin{enumerate}
		\item Fix $q = k^{1+\beta}$ and $z = \frac{q}{2 \ln n}$. Using Lemma~\ref{lemma:groupsPartition}, partition the
		set of edges of $G$ into $z$ groups $\{E_1, E_2, \dots, E_z\}$
		\item Similar to the hardness construction in Section~\ref{sec:eth-hardness}, all the request pairs and obstacles are
					contained in an enclosing rectangle $\bbox$ with bottom left corner $(0, 0)$ and top-right corner $(z, 4)$.

		\item For every $v_i \in V$, add an obstacle $S_i$ to $\cs$.
					Initially, all obstacles are horizontal line segments occupy the part of $x$-axis from $x=0$ to $x=z$.

		\item Define two set of horizontal lines $\ell_1^h : y = 1 ~+~ \frac{h}{(|E| + 1)}$ 
					and $\ell_2^h : y = 3 ~+~ \frac{h}{(|E| + 1)}$ to be a horizontal line that will serve as
					\emph{guardrails} for obstacle growth corresponding to edge $e_h \in E$. Here $1 \leq h \leq |E|$.
					We will refer to the group $\ell_1^h$,  $\ell_2^h$ lines as \emph{$\ell_1$-channel} and 
					\emph{$\ell_2$-channel} respectively.

		\item For each group $E_r$, define a \emph{request pair block} $B_r$ which is a unit-width 
					sub-rectangle of $\bbox$ bounded by vertical sides
					$x={r-1}$ and $x=r$.
					Let $\mid_r = (r-\frac{1}{2})$ and add the pair of points $a_r = (\mid_r, \frac{5}{2})$ and $a_r' = (\mid_r, \frac{1}{2})$ 
					to $A$. These points will be contained in block $B_r$. 
					
					Now for every edge $e_h = (v_i, v_j)\in E_r$ with $i < j$, we grow the obstacles along $\ell_1, \ell_2$-channels as follows.
					(See also Figure~\ref{fig:apx-hardness-example}.)
					\begin{itemize}
						\item Grow the obstacle $S_i$ corresponding to vertex $v_i$ vertically along left boundary $x= r-1$ of $B_r$ until $y=4$.
									Similarly grow $S_j$ along right boundary $x=r$ of $B_r$ until $y=4$.
						\item Moving along the horizontal line $\ell_1^h$ from \emph{left to right}, extend obstacle $S_i$ from
									$x=r-1$ to $x= \mid_r$. Repeat the same for $\ell_2^h$.
						\item Similarly, moving along the horizontal line $\ell_1^h$ from \emph{right to left}, extend obstacle $S_j$ from
						      $x=r$ to $x= \mid_r$. Repeat the same for $\ell_2^h$.
					\end{itemize}
	\end{enumerate}

	\begin{figure}[htb!]
	\centering
	\includegraphics{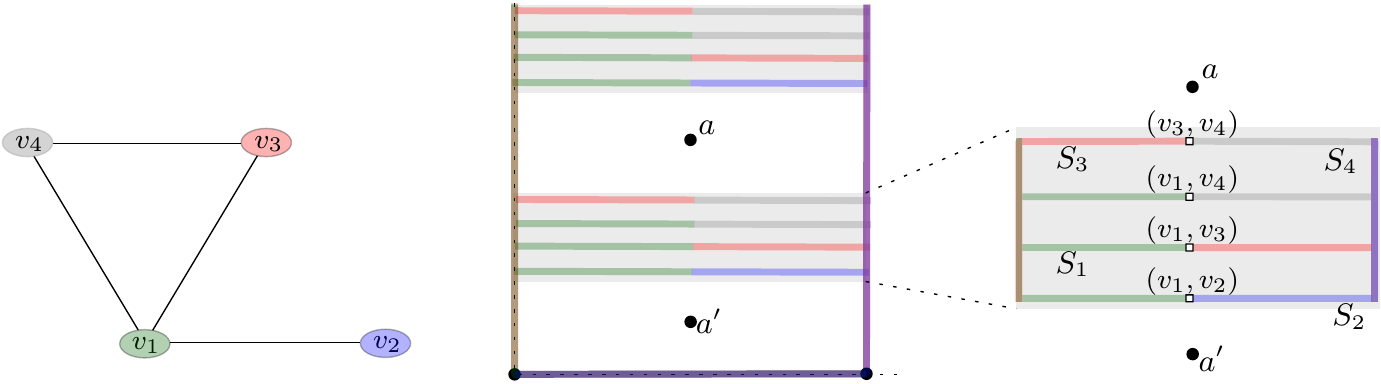}
	\caption{An group of edges $E_1$ and the resulting \pointsep request pair block $B_1$. The $\ell_1$-channel
					is shown enlarged in the rightmost figure.
					As an example, observe that point pair $(a, a')$ is separated if obstacles 
				$S_1, S_2$ are selected (because $(v_1, v_2) \in E_1$) 
				but not separated if obstacles $S_2, S_3$ are selected (because $(v_2, v_3) \not\in E_1$).} 
		\label{fig:apx-hardness-example}
	\end{figure}

	\begin{lemma}
		\label{lemma:edge-barrier}
		Let $\cs^* \subseteq \cs$ be a solution to the \pointsep instance $(\cs, A)$ constructed above.
		Then all point pairs in $A$ are separated if and only if 
		for every request pair block $B_r$, there exists two obstacles $S_i, S_j \in \cs^*$ such 
		that $(v_i, v_j)$ is an edge assigned to group $E_r$.
	\end{lemma}
	\begin{proof}
			For the forward direction, suppose we start moving vertically in block $B_r$ 
			along $x = \mid_r$ starting from $a_r'$ towards $a_r$.
			Before we reach point $a_r$, we must cross the lines $\ell^1_h$ for all $h$ such that $ e_h \in E_r$.
			Whenever we arrive at $\ell^1_h$ which is the guardrail corresponding to edge $e_h = (v_i, v_j)$, 
			if either $S_i \not\in \cs^*$ or $S_j \not\in \cs^*$, then we can \emph{cross over} $\ell^1_h$ without 
			intersecting an obstacle in $\cs^*$ by shifting infinitesimally to the left (or right) from $x=\mid_r$.
			Since $\cs^*$ separates $a_r, a_r'$, there must be some $e_h = (v_i, v_j)$ with $i < j$, 
			such that both $S_i, S_j \in \cs^*$.

			For the other direction, if obstacles $S_i, S_j \in \cs^*$ such that $(v_i, v_j) \in E_r$, 
			then the union of $S_i, S_j$ forms a closed curve enclosing both $a$ and $a'$ and therefore 
			separates $a, a'$ from each other as well as from other points in $A$.
	\end{proof}

Using the discussion preceding Lemma~\ref{lemma:distinguishing-ratio}, 
we can obtain a lowerbound for the distinguishing ratio $\rho $ of the above reduction as follows.
\begin{lemma}
	\label{lemma:dense-random-values}
		Let $(\cs, A)$ be the resulting \pointsep instance obtained by applying the above reduction to a graph $G$.
		Then we have distinguishing ratio:
		\begin{enumerate}
				\item $\rho ~~\geq~~ \min\left\{k^\beta,~~ \sqrt{k^{\beta - 1} \cdot n^{1 - \alpha}}\right\}$ in terms of $n, k$
				\item $\rho ~~\geq~~ \frac{\min\left\{q,~~ \sqrt{q \cdot n^{1 - \alpha}}\right\}}{q^{1/(\beta+1)}}$ in terms of $n, q$.
		\end{enumerate}
\end{lemma}

\begin{proof}
	We have the following two cases for the instance $(\cs, A)$ depending on graph $G$.
	\begin{itemize}
	\item \emph{$G$ contains a subgraph on $k$ vertices and $q=k^{\beta+1}$ edges}.
		Let $E^*$ be the set of these edges. Using Lemma~\ref{lemma:groupsPartition}, it follows that 
		every group $E_r$ contains an edge $e_h \in E^*$. Using the obstacles corresponding to vertices in $E^*$
		and applying Lemma~\ref{lemma:edge-barrier},
		we obtain a set of at most $k$ obstacles that separate the request pair $(a_r, a_r')$ in every block $B_r$.
		Therefore, the number of obstacles used in this case $x_d^* \leq k$.
	
	\item \emph{$G$ is drawn from $G(n, p)$ with $p = n^{\alpha - 1}$}.
		From Lemma~\ref{lemma:edge-barrier}, it follows that to separate $(a_r, a_r')$ 
		in any block $B_r$, any solution
		must select both obstacles corresponding to at least one edge in $B_r$. 
		Choosing one edge from each block, we obtain a subgraph of $G$ with $z$ edges.
		Applying Lemma~\ref{lemma:random-subgraph} on this subgraph and observing 
		that $z = \tilde{\Omega}(q)$ gives the number of obstacles used in this case 
		$x_r^* \geq \tilde{\Omega}(\min\{q, \sqrt{(q/p)}\})$.
	\end{itemize}
		Taking the ratio of solution sizes in both cases and substituting the values $p = n^{\alpha - 1}$ and $q=k^{\beta+1}$ , we obtain:
		\begin{align*}
			\rho ~=~ \frac{x_r^*}{x_d^*} ~\geq~ \frac{\min\left\{k^{\beta+1},~~ \sqrt{\left(\frac{k^{\beta+1}}{n^{\alpha - 1}}\right)}\right\}}{k} ~=~
							\min\left\{k^\beta,~~ \sqrt{k^{\beta - 1} \cdot n^{1 - \alpha}}\right\}
		\end{align*}
		Similarly, in terms of $q, n$, we obtain the following:
		\begin{align*}
			\rho ~=~ \frac{x_r^*}{x_d^*} ~\geq~ \frac{\min\left\{q,~~ \sqrt{q \cdot n^{1 - \alpha}}\right\}}{k} ~=~ 
			\frac{\min\left\{q,~~ \sqrt{q \cdot n^{1 - \alpha}}\right\}}{q^{1/(\beta+1)}}
		\end{align*}

\end{proof}

We will now fix the choice of parameters $\alpha, \beta, k$ such that they satisfy 
the requirements of Conjecture~\ref{conj:dvr} and obtain a bound on the distinguishing
ratio in terms of number of obstacles $|\cs| = n$ and number of points $|A| = 2z = \tilde{\Theta}(q)$.
The parameters are carefully chosen so that the maximize the distinguishing ratio and therefore
obtain best possible lowerbound on the hardness of approximation.

\begin{lemma}
	Assuming \dvrconj, a \pointsep instance $(\cs, A)$ cannot be approximated to a factor better
	than $|A|^{1/2 - \eps}$ in polynomial time, for any $\eps > 0$.
\end{lemma}
\begin{proof}
	Let $\alpha = 1 - \eps$ and $\beta = \alpha - \eps$ and $q = n^{1-\alpha} = n^\eps$.
	Since, $k^{\beta + 1} = q$, we have $ k^{\beta + 1} =  n^\eps < n^{(1+\alpha)/2}$.
	Therefore the parameters $\alpha, \beta, k$ satisfy the conditions of Conjecture~\ref{conj:dvr}.
	Since $q = n^{1-\alpha}$, we have $\min\{q,~~ \sqrt{q \cdot n^{1-\alpha}}\} = q$.
	Substituting this to the equation
	for $\rho$ in terms of $n, q$ from Lemma~\ref{lemma:dense-random-values}, we obtain:
	\begin{align*}
		\rho ~~\geq~~ \frac{q}{q^{1/(\beta+1)}} ~~=~~ q^{\frac{\beta}{\beta+1}} ~~=~~ q^{\frac{(1 - 2\eps)}{2- 2\eps}} 
		~~=~~ q^{\frac{1 - \eps}{2- 2\eps} - \frac{\eps}{2- 2\eps}} ~~=~~ q^{1/2 - \eps'}
	\end{align*}
	where $\eps' = \frac{\eps}{2-2\eps}$. Since $|A| = \tilde{\Theta}(q)$, applying 
	Lemma~\ref{lemma:distinguishing-ratio}, we achieve the claimed bound.
\end{proof}

\begin{lemma}
	Assuming \dvrconj, a \pointsep instance $(\cs, A)$ cannot be approximated to a factor better
	than $|\cs|^{3-2\sqrt{2}-\eps}$ in polynomial time, for any $\eps > 0$.
\end{lemma}
\begin{proof}
	For this case, we set both $\alpha = \sqrt{2} - 1$ and $k = n^{\sqrt{2} - 1}$.
	With $\beta = \alpha - \eps$, we have $k^{\beta+1} \leq k^{\alpha+1} = n^{2 - \sqrt{2}} < n^{(1+\alpha)/2}$
	which satisfies the requirements of Conjecture~\ref{conj:dvr}.

	Therefore, we have:
	\begin{align*}
		k^{\beta} ~&=~ n^{(\sqrt{2} - 1) \cdot (\sqrt{2} - 1 - \eps)} = n^{3-2\sqrt{2} - \eps'}  &&\text{for some $\eps' > 0$} \\
		\sqrt{k^{\beta - 1} \cdot n^{1 - \alpha}} ~&=~ \left(n^{(\sqrt{2} - 1)(\beta -1)} \cdot n^{(\sqrt{2} - 1)\sqrt{2}} \right)^{1/2} 
					~=~  n^{\frac{(\sqrt{2} - 1)(2\sqrt{2} - 2 - \eps)}{2}} \\
					~&=~  n^{3-2\sqrt{2} - \eps''}  &&\text{for some $\eps'' > 0$}
	\end{align*}
	Substituting this to the equation
	for $\rho$ in terms of $n, k$ from Lemma~\ref{lemma:dense-random-values} and applying Lemma~\ref{lemma:distinguishing-ratio},
	we achieve the claimed bound.
\end{proof}

We conclude with the main result for this section.
\begin{theorem}
	\label{thm:apx-hardness}
	Assuming \dvrconj~\cite{chlamtavc2017minimizing}, one cannot approximate \pointsep within ratio
	$n^{3-2\sqrt{2} - \eps}$ or $m^{1/2 - \eps}$ in polynomial time, for any $\eps > 0$, 
	where $n$ is the number of obstacles and $m$ is the number of points.
\end{theorem}


\bibliography{refs}

\begin{thebibliography}{10}

\bibitem{agarwal2005log}
A.~Agarwal, M.~Charikar, K.~Makarychev, and Y.~Makarychev.
\newblock {O ($\sqrt {\log n}$) approximation algorithms for Min UnCut, Min
  2CNF Deletion, and directed cut problems}.
\newblock In {\em Proc. of 37th STOC}, pages 573--581, 2005.

\bibitem{AgarwalCMM05}
Amit Agarwal, Moses Charikar, Konstantin Makarychev, and Yury Makarychev.
\newblock O(sqrt(log n)) approximation algorithms for min uncut, min 2cnf
  deletion, and directed cut problems.
\newblock In {\em Proceedings of the 37th Annual {ACM} Symposium on Theory of
  Computing, Baltimore, MD, USA, May 22-24, 2005}, pages 573--581, 2005.

\bibitem{balister2009trap}
Paul Balister, Zizhan Zheng, Santosh Kumar, and Prasun Sinha.
\newblock Trap coverage: Allowing coverage holes of bounded diameter in
  wireless sensor networks.
\newblock In {\em IEEE INFOCOM 2009}, pages 136--144. IEEE, 2009.

\bibitem{bksvCGTA}
Sayan Bandyapadhyay, Neeraj Kumar, Subhash Suri, and Kasturi Varadarajan.
\newblock Improved approximation bounds for the minimum constraint removal
  problem.
\newblock {\em Computational Geometry}, 90:101650, 2020.

\bibitem{BeregK09}
Sergey Bereg and David~G. Kirkpatrick.
\newblock Approximating barrier resilience in wireless sensor networks.
\newblock In {\em Proc. of 5th ALGOSENSORS}, volume 5804, pages 29--40, 2009.

\bibitem{cabello2016complexity}
S.~Cabello and P.~Giannopoulos.
\newblock The complexity of separating points in the plane.
\newblock {\em Algorithmica}, 74(2):643--663, 2016.

\bibitem{ChanK12}
David Yu~Cheng Chan and David~G. Kirkpatrick.
\newblock Approximating barrier resilience for arrangements of non-identical
  disk sensors.
\newblock In {\em Proc. of 8th {ALGOSENSORS}}, pages 42--53, 2012.

\bibitem{ChanK14}
David Yu~Cheng Chan and David~G. Kirkpatrick.
\newblock Multi-path algorithms for minimum-colour path problems with
  applications to approximating barrier resilience.
\newblock {\em Theor. Comput. Sci.}, 553:74--90, 2014.

\bibitem{chlamtavc2017minimizing}
Eden Chlamt{\'{a}}c, Michael Dinitz, and Yury Makarychev.
\newblock Minimizing the union: Tight approximations for small set bipartite
  vertex expansion.
\newblock In {\em Proc. of 28th SODA}, pages 881--899, 2017.

\bibitem{Eiben2017}
E.~Eiben and I.~Kanj.
\newblock How to navigate through obstacles?
\newblock In {\em Proc. of 45th ICALP}, 2018.

\bibitem{eiben2018improved}
Eduard Eiben, Jonathan Gemmell, Iyad~A. Kanj, and Andrew Youngdahl.
\newblock Improved results for minimum constraint removal.
\newblock In {\em Proc. of 32nd {AAAI}}, pages 6477--6484, 2018.

\bibitem{EibenK20}
Eduard Eiben and Iyad Kanj.
\newblock A colored path problem and its applications.
\newblock {\em {ACM} Trans. Algorithms}, 16(4):47:1--47:48, 2020.

\bibitem{mcp-fpt}
Eduard Eiben and Daniel Lokshtanov.
\newblock Removing connected obstacles in the plane is {FPT}.
\newblock In {\em Proc. of 36th {SoCG}}, volume 164, pages 39:1--39:14, 2020.

\bibitem{lavalle}
Lawrence~H. Erickson and Steven~M. LaValle.
\newblock A simple, but {NP-Hard}, motion planning problem.
\newblock In {\em Proc. of 27th {AAAI}}, 2013.

\bibitem{constApxPseudodisks}
Matt Gibson, Gaurav Kanade, and Kasturi Varadarajan.
\newblock On isolating points using disks.
\newblock In {\em Algorithms -- ESA 2011}, pages 61--69, 2011.

\bibitem{ImpagliazzoPZ01}
Russell Impagliazzo, Ramamohan Paturi, and Francis Zane.
\newblock Which problems have strongly exponential complexity?
\newblock {\em J. Comput. Syst. Sci.}, 63(4):512--530, 2001.

\bibitem{kedem1986union}
Klara Kedem, Ron Livne, J{\'a}nos Pach, and Micha Sharir.
\newblock On the union of jordan regions and collision-free translational
  motion amidst polygonal obstacles.
\newblock {\em Discrete \& Computational Geometry}, 1(1):59--71, 1986.

\bibitem{KormanLSS18}
Matias Korman, Maarten L{\"{o}}ffler, Rodrigo~I. Silveira, and Darren Strash.
\newblock On the complexity of barrier resilience for fat regions and bounded
  ply.
\newblock {\em Comput. Geom.}, 72:34--51, 2018.

\bibitem{kratsch2020representative}
Stefan Kratsch and Magnus Wahlstr{\"o}m.
\newblock Representative sets and irrelevant vertices: New tools for
  kernelization.
\newblock {\em Journal of the ACM (JACM)}, 67(3):1--50, 2020.

\bibitem{KumarLA05}
Santosh Kumar, Ten{-}Hwang Lai, and Anish Arora.
\newblock Barrier coverage with wireless sensors.
\newblock {\em Wirel. Networks}, 13(6):817--834, 2007.

\bibitem{lee2016separators}
James~R. Lee.
\newblock Separators in region intersection graphs.
\newblock In {\em Proc. of 8th {ITCS}}, volume~67, pages 1--8, 2017.

\bibitem{oct-fpt}
Daniel Lokshtanov, NS~Narayanaswamy, Venkatesh Raman, MS~Ramanujan, and Saket
  Saurabh.
\newblock Faster parameterized algorithms using linear programming.
\newblock {\em ACM Transactions on Algorithms (TALG)}, 11(2):1--31, 2014.

\bibitem{marx2007can}
D{\'a}niel Marx.
\newblock Can you beat treewidth?
\newblock In {\em 48th Annual IEEE Symposium on Foundations of Computer Science
  (FOCS'07)}, pages 169--179. IEEE, 2007.

\bibitem{marx2012tight}
D{\'a}niel Marx.
\newblock A tight lower bound for planar multiway cut with fixed number of
  terminals.
\newblock In {\em International Colloquium on Automata, Languages, and
  Programming}, pages 677--688. Springer, 2012.

\bibitem{mcp2021}
Saket~Saurabh Neeraj~Kumar, Daniel~Lokshtanov and Subhash Suri.
\newblock A constant factor approximation for navigating through connected
  obstacles in the plane.
\newblock In {\em Proc. 32nd SODA}, 2021.

\end{thebibliography}


\end{document}